\documentclass[sigconf]{acmart}

\usepackage{booktabs} 
\usepackage{arydshln} 
\usepackage{graphicx}
\graphicspath{ {images/} }
\usepackage{pdfpages}
\usepackage{blindtext}
\usepackage{hyperref}
\usepackage{subcaption}
\usepackage{amsmath}
\usepackage{bm} 
\usepackage{amsmath, mathrsfs,amsfonts,amssymb}
\usepackage{wasysym} 
\usepackage{xfrac} 
\DeclareMathOperator*{\argmin}{arg\,min}
\DeclareMathOperator*{\argmax}{arg\,max}
\usepackage{savesym} 
\savesymbol{openbox}
\usepackage{amsthm} 
\usepackage{thmtools}
\usepackage{cleveref}
\usepackage{commath}
\usepackage[export]{adjustbox}

\setcopyright{rightsretained}

\begin{document}
\title{Task-driven sampling of attributed networks}

\author{Suhansanu Kumar}
\affiliation{%
  \institution{University of Illinois at Urbana-Champaign}
  \city{Urbana}
  \state{Illinois}
  \postcode{60801}
}
\email{skumar56@illinois.edu}

\author{Hari Sundaram}
\affiliation{%
  \institution{University of Illinois at Urbana-Champaign}
  \city{Urbana}
  \state{Illinois}
  \postcode{60801}
}
\email{hs1@illinois.edu}

\begin{abstract}
This paper introduces new techniques for sampling attributed networks to support standard Data Mining tasks. The problem is important for two reasons. First, it is commonplace to perform data mining tasks such as clustering and classification of network attributes (attributes of the nodes, including social media posts). Furthermore, the extraordinarily large size of real-world networks necessitates that we work with a smaller graph sample. Second, while random sampling will provide an unbiased estimate of content, random access is often unavailable for many networks. Hence, network samplers such as Snowball sampling, Forest Fire, Random Walk, Metropolis-Hastings Random Walk are widely used; however, these attribute-agnostic samplers were designed to capture salient properties of network structure, not node content. The latter is critical for clustering and classification tasks. There are three contributions of this paper. First, we introduce several attribute-aware samplers based on Information Theoretic principles. Second, we prove that these samplers have a bias towards capturing new content, and are equivalent to uniform sampling in the limit. Finally, our experimental results over large real-world datasets and synthetic benchmarks are insightful: attribute-aware samplers outperform both random sampling and baseline attribute-agnostic samplers by a wide margin in clustering and classification tasks.
\end{abstract}
\maketitle


\section{Introduction}\label{sec:intro}
In this paper, we propose new sampling algorithms for attributed networks. By network attributes, we specifically mean content attributes such as gender, location etc. that are distinct from attributes arising from network structure (e.g. node degree, clustering coefficient).

Sampling networks with the aim of improving data mining task performance is important. Classification~\cite{sorkhoh2008classification}, community discovery~\cite{yang2013community} as well as  clustering of nodes into functional groups \cite{handcock2007model} using node content are familiar data mining tasks on networks. The extraordinarily large size of real-world networks (e.g. Facebook has over a billion nodes) necessitates that we work with a smaller graph sample. To sample, most researchers use well known graph sampling methods such as snowball sampling, or stochastic samplers such as Random Walk, Forest Fire~\cite{leskovec2006sampling} and Metropolis-Hastings Random Walk (MHRW)~\cite{hubler2008metropolis}. There is an implicit assumption that these samplers are ``good enough'' to form representative samples for their task. However, much of the early work on network sampling focused on preserving the structural properties of the network in the sample, not to discover patterns in the node content attributes.

Sampling uniformly at random, over the population is the standard for developing an unbiased estimate of the attribute value distribution and associated statistics including mean and variance. Random access to nodes in a graph is not always available and social networks including Facebook and Pinterest actively prevent such access. Thus we use \textit{link-trace} samplers such as Snowball sampling or Random Walk, where each node added to the sample has a neighbor in the current sample. Indeed MHRW~\citep{hubler2008metropolis} was designed so that the stationary distribution over the graph is uniform. The challenge with MHRW is that for \textit{finite samples} the probability of visiting each node is not uniform.

\textit{If our goal is to cluster or to classify the node content, can we do ``better'' than uniform sampling over the graph?} We motivate this question by an illustrative example.~\Cref{fig:motivation} shows a training set comprising two classes, A and B where each class represented by a set of two dimensional samples; each class is of a different size. Assume that we are trying to learn a discriminative classifier (e.g. SVM). We know from standard Machine Learning theory that the most informative samples to build the SVM lie at the boundary of each class. Uniform sampling of each class A, B, will have uninformative samples picked away from class boundary. In contrast, we would ideally like to pick points near the class boundary. Thus we should expect that for the same \textit{finite} sample size $N$, uniform sampling should have higher generalization error than for those sets of size $N$ containing samples primarily near the boundary for each class. This example motivated our idea of obtaining surprising or extremal samples for clustering and classification tasks.

We specifically look at three tasks---data characterization, clustering and classification. Our contributions are as follows:

\begin{enumerate}
    \item We propose several new \textit{link-trace} samplers grounded in Information Theory. These ``Information Expansion'' samplers seek out previously unseen content samples rapidly covering the attribute range.
    \item We characterize the bias of these Information theoretic samplers, and prove two lemmas: they are biased towards collecting nodes with attribute values absent from the current sample; asymptotic behavior tends towards uniform distribution. In practice, this means that for small sample sizes, information expansion samplers perform stratified sampling, covering more informative samples.
\end{enumerate}

We have interesting results for all three tasks. In all three cases---characterization, clustering and classification---attribute-aware samplers substantially outperform baselines (BFS, RW, MHRW). For data characterization task only, uniform random sampling has the best raw performance (KS statistic). However, this is statistically indistinguishable from content aware link-trace sampling. For clustering task for example, there is an average of 45\% improvement over uniform sampling at a sample size of 5\% for real-world data sets. The improvements are more significant at smaller sample sizes. For example, in Patent network, 5.7\% of the patents sampled via proposed samplers achieve the same clustering performance as 10\% of the patents collected uniformly from the dataset. This amounts to a saving of over 100K nodes in sampling.


\setlength{\belowcaptionskip}{-10pt}
\begin{figure}[t!]
  \centering
\includegraphics[width=0.75\linewidth]{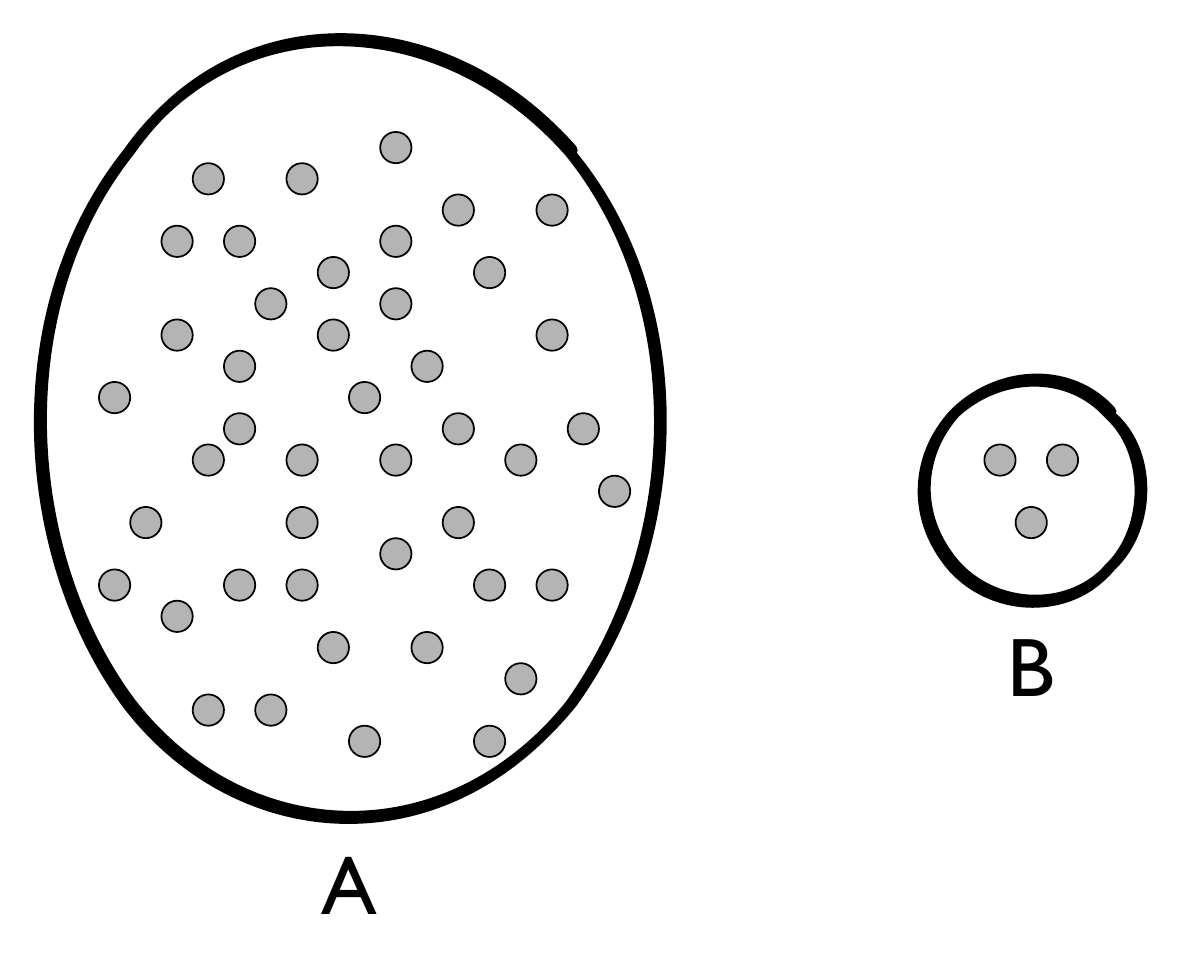}
\caption{Figure shows two classes A and B in the 2D plane. Content nodes at the boundary are the most informative nodes for classification. In support vector machines, they are also called as support vectors. Uniform sampling has two challenges: it collects non-informative samples away from the class boundary, and will under sample or even miss small classes.}
\label{fig:motivation}
\end{figure}

 Furthermore, while samplers such as MHRW have been shown to be asymptotically equivalent to uniform sampling the graph, however,  finite sample statistics of MHRW reveal that for a finite sample, the probability of visiting a node is not uniform over the network.

We show via a stylized example the differences between two link-trace samplers---MHRW and  Information eXpansion Sampling (IXS). \Cref{fig:clarifyingexample1} shows an attributed network with a strong community structure and having a single discrete attribute; the different colors in the graph refer to different attribute values. The two sub-figures show a single trace of size equal to 10\% of the network size, of two algorithms (MHRW, Information eXpansion Sampler) starting from the same seed node (marked in red). The sampled nodes are marked with a dark black ring, and the edges of the induced subgraph are colored black. As can be seen from~\Cref{fig:clarifyingexample1}, MHRW gets ``stuck'' in a small section of the network even though it has an asymptotically optimal performance characteristic. For a small sample size, MHRW has a known bias towards low-degree nodes. Note that IXS with its bias towards capturing new attribute values is much more efficient at covering the attribute space. In a similar vein illustrated through \Cref{fig:clarifyingexample2}, we can show that IXS is superior to XS~\cite{maiya2011benefits} (an attribute-agnostic sampler that performs well with networks with community structure) for dis-assortative attributed-networks with poor community structure. In summary, our proposed Information eXpansion Sampling algorithm expands rapidly in the content space when there are attributes to be discovered, but more like a random walker if the information in the network neighborhood of the sample fails to provide guidance.
\setlength{\belowcaptionskip}{-10pt}
\begin{figure}
\begin{subfigure}{.25\textwidth}
  \centering
  \includegraphics[width=1\linewidth]{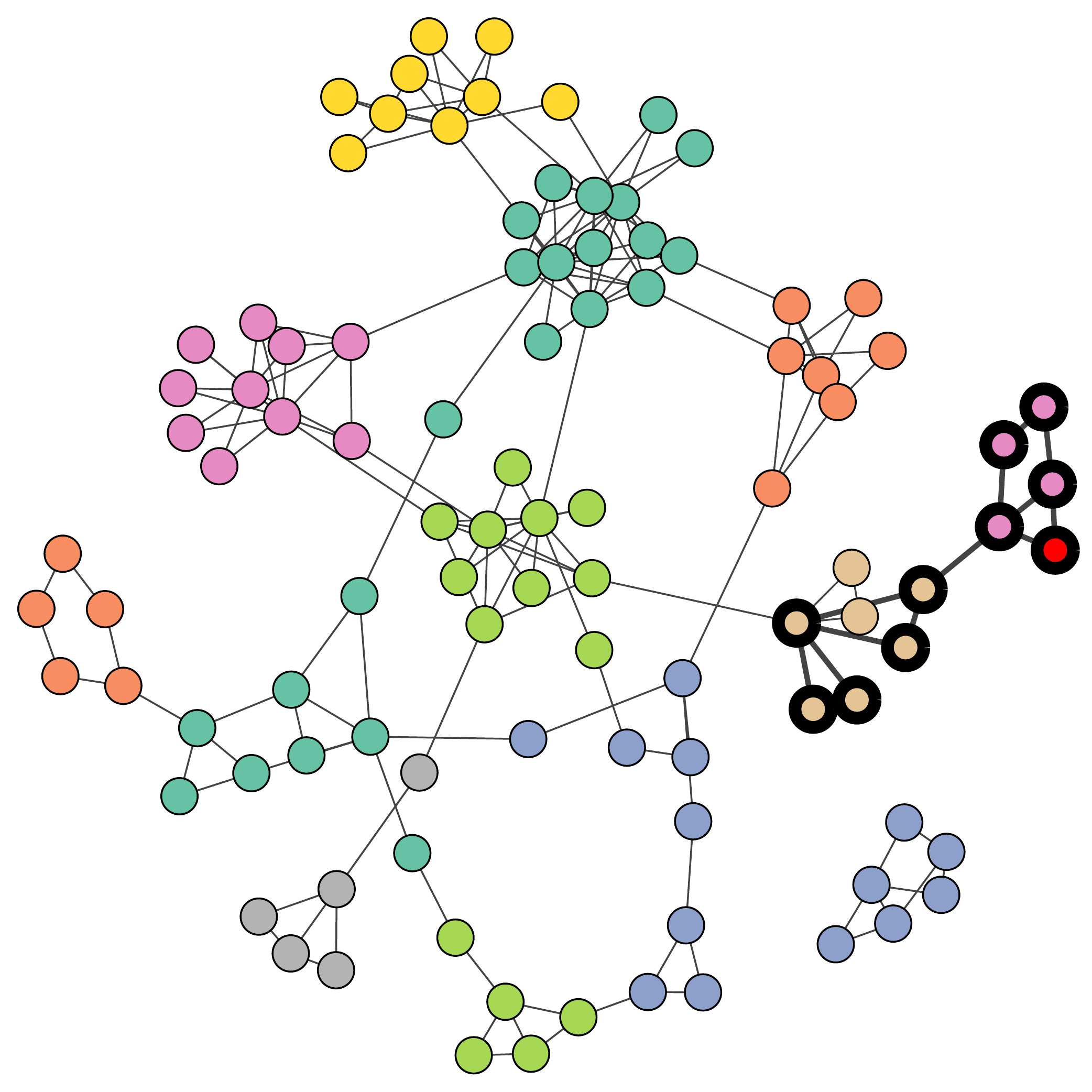}
  \caption{a MHRW run}
  \label{fig:sfig1}
\end{subfigure}%
\begin{subfigure}{.25\textwidth}
  \centering
  \includegraphics[width=1\linewidth]{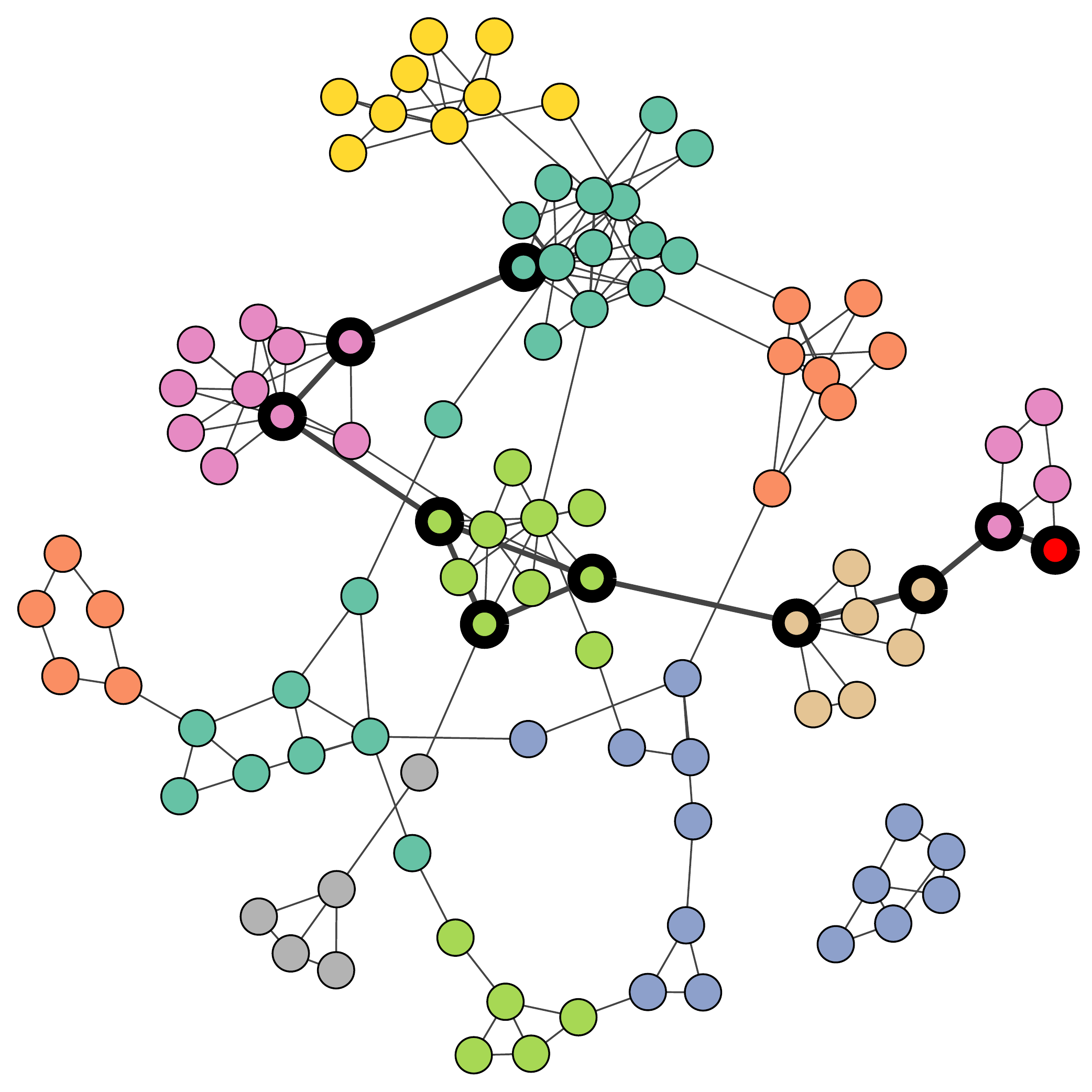}
  \caption{an IXS run}
  \label{fig:sfig2}
\end{subfigure}
\caption{The Figure shows a network where the different colors represent the different attribute values; sampled nodes with darker concentric rings; the target sample size is 10\%. Subplot (a) shows MHRW, a random walk based sampler,  getting stuck in a local part of the network, due to its bias. Subplot (b) shows data aware sampler IXS overcoming this bottleneck due to its bias for new information.}
\label{fig:clarifyingexample1}
\end{figure}
\setlength{\belowcaptionskip}{0pt}

The rest of this paper is organized as follows. In the next section we formally define the sampling problem. In~\Cref{sec:How to Sample}, we discuss attribute-agnostic and attribute-aware papers and introduce our information expansion based samplers. In the three following sections, we present results for synthetic and real-world datasets for baseline and our attribute-aware samplers for data characterization, clustering and classification tasks. In~\Cref{sec:relatedwork}, we discuss prior work and in~\Cref{sec:Limitations}, we discuss limitations. We present our conclusions in~\Cref{sec:conclusion}.

 \begin{figure}
 \begin{subfigure}{.25\textwidth}
   \centering
   \includegraphics[width=1\linewidth]{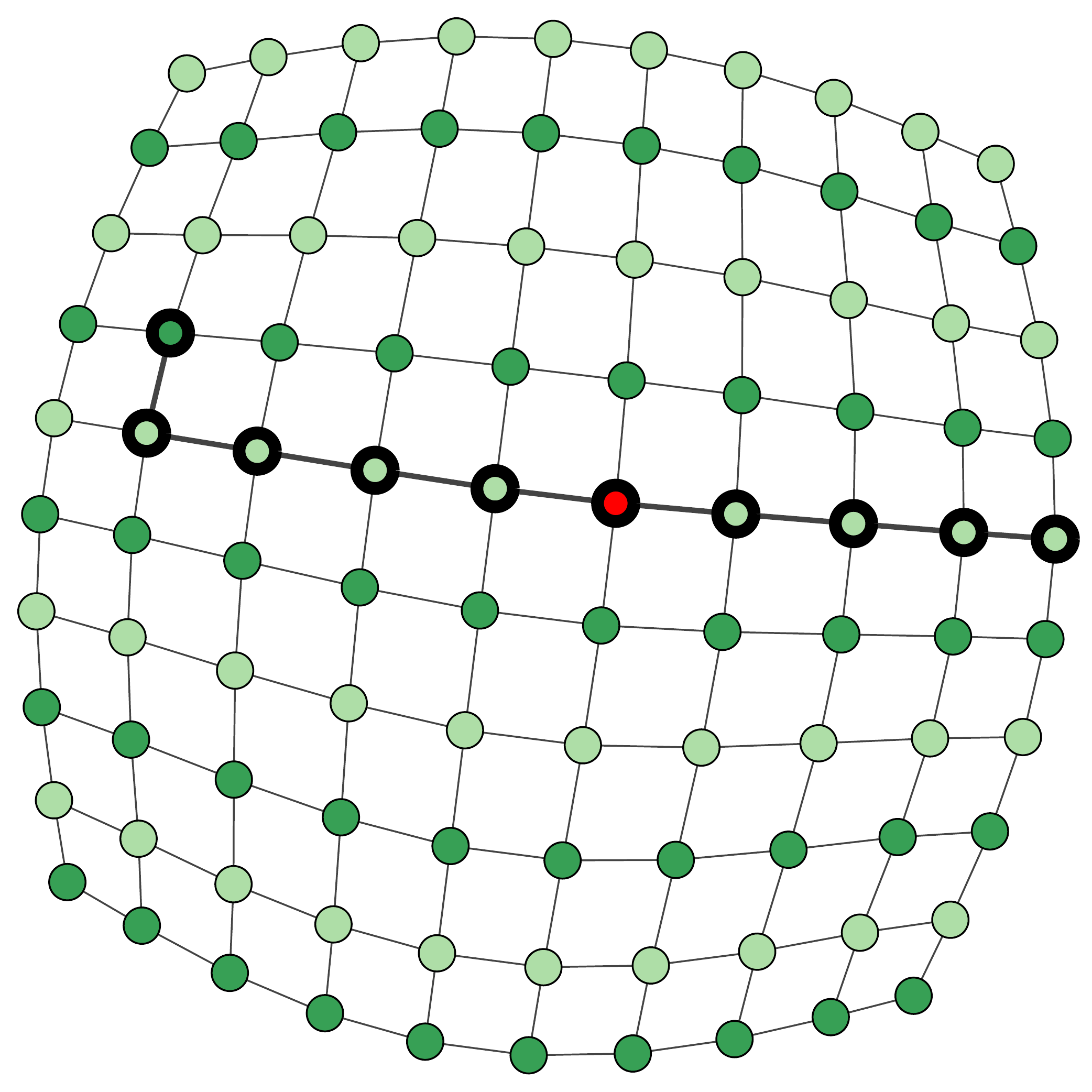}
   \caption{a XS run}
   \label{fig:sfig1}
 \end{subfigure}%
 \begin{subfigure}{.25\textwidth}
   \centering
   \includegraphics[width=1\linewidth]{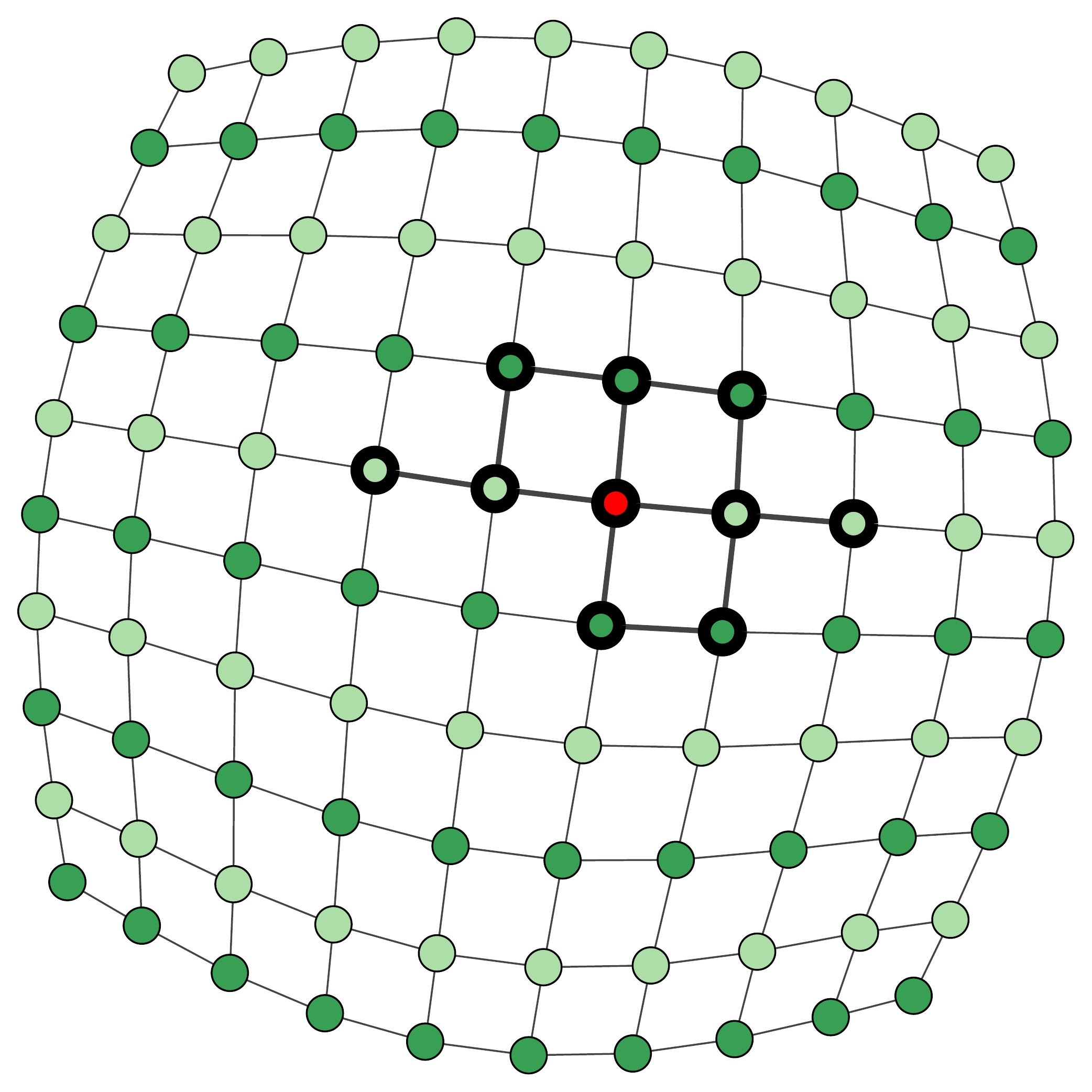}
   \caption{an IXS run}
   \label{fig:sfig2}
 \end{subfigure}
 \caption{The Figure shows a network where the two shades of green represents the two different attributes of the nodes. Subplot (a) shows expansion sampling (XS) behaves poorly when the attribute color has low assortativity of 0. Subplot (b) shows attribute awareness while sampling can alleviate this problem. The two shades of green represent the different attribute values.}
 \label{fig:clarifyingexample2}
 \end{figure}


\section{Preliminaries}
\label{sec:Preliminaries}

In this section, we shall first define the notation used in the paper (\Cref{tab:notationstable}), then we shall formally define the task-driven (or purposeful) sampling problem in~\Cref{sub:Problem Statement}. We shall conclude the section with a discussion of the real-world datasets used as well as the mechanisms to create realistic attributed datasets.

%


\begin{table}[h]
\center
    \caption{Notation used in the paper}\label{tab:notationstable}
\begin{tabular}{@{}c l@{}} \toprule
    \textbf{Symbol} & \textbf{Definition}\\ \midrule
$G$  & Network, $ G = (V, E)$ \\
$V$  & Vertices of the network \\
$E$  & Edges of the network, $E \subseteq V \times V$ \\
$v$  & vertex(node) in network \\
$d_v$  & degree of node $v$ in network \\
$d_V  $  & mean degree of nodes $v \in V$: $d_V= \frac{1}{|V|}\sum_v d_v$ \\
$A(v)$  & Attribute vector or content of a node $v$ \\
$\mathbb{S}$  & Sampled nodes \\
$N(\mathbb{S})$  & Frontier node set or neighborhood of set $\mathbb{S}$ \\
 & $\{w \in V \setminus \mathbb{S} : \exists v \in \mathbb{S} : (v, w) \in E\}$\\
$\Delta(v)$  & Unexplored neighbors of a node $v$ \\
 & $\{w \in N(v) \setminus (\mathbb{S} \cup N(\mathbb{S})) \}$\\
$\mathbb{C}$  & Set of content clusters in G. \\
$k$  & Number of content clusters in G. \\ \bottomrule
\end{tabular}

\end{table}

\subsection{Problem Statement}
\label{sub:Problem Statement}

Assume that we have a graph $G = (V, E)$, where each node has $m$ content attributes (e.g. gender, location, etc.) and that we have a task $F$ that performs operations using an attributed graph as input. The task $F$ produces an output: a data characterization in the form of a scalar (e.g. mean of an attribute), or a vector (e.g. distribution of gender);  a mapping (e.g. assignment of each node to a cluster); a function (e.g. a classifier that operates on further input).

Further assume that the input to $F$ is a sample $\mathbb{S}$ of size $z \ll |V|$. The sample $\mathbb{S}$ can be obtained through a variety of ways, for example, random sampling of nodes. However, in practice  random access to the nodes is rare. Instead, most network sampling mechanisms are \textit{link-trace} samplers. We define link trace sampling in a manner similar to~\citep{maiya2011benefits} as follows. Given an integer $z$ and an initial seed node $v \in V$ to which $\mathbb{S}$ is initialized, a link trace sampler $L$ adds node $v$ to $\mathbb{S}$ such that there exists a node $w \in \mathbb{S}$  where $(w,v) \in E$. The sampler stops when $|\mathbb{S}|=z$.

Link-trace samplers yield connected components since each new addition must lie in $N(\mathbb{S})$, the neighborhood of $\mathbb{S}$. We ran the link-trace samplers on largest component of the graph leading to a sample collection of $z$ nodes from the $|V|$ nodes of the underlying network. 

This sample $\mathbb{S}$ is associated with an induced subgraph $G'_{z}=(V_{z},E_{z})$. Thus for a given sample size $z$, seed nodes $\theta$ and a task $F$, the goal is to find an optimal link-trace sampler $L^*$ such that,
\begin{align}\label{eq:objective}
L^* = \argmin_L \mathbb{E}_\theta  D \left ( F(G), F(G'_{z}; L, \theta) \right)
\end{align}

The  function $D$ measure the distance between the outputs for task $F$ in the ideal case with the entire graph $G$ as input against the case when the sampled graph $G'_{z}$ is used as input. The graph $G'_{z}$ is parameterized by the sampler $L$ and the seed set $\theta$. The distance measure is task dependent: $D$ could be just the absolute difference in values, say when $F$ is computing the mean of an attribute, the KS statistic in the case when $F$ computes a distribution, or Normalized Mutual Information when $F$ performs clustering.

There may be two sources of randomness involved in sampling, depending on the type of sampler. The first source is the location of the seed set and the second may be the sampler itself. For stochastic samplers such as Random-Walk, given the same seed, every run of the algorithm will produce a different sample. The expectation $\mathbb{E}$ is over both sources of information, although~\Cref{eq:objective} only refers to the expectation over $\theta$. Thus,~\Cref{eq:objective} simply says that we should select the link trace sampler with minimum distance to the ideal case, averaged over different seed sets and over different link traces.

\subsection{Datasets}
\label{sub:Datasets}
\begin{table}[b]
\center
\caption{Network statistics. N: number of nodes, E: number of edges, DS: number of discrete attributes, CT: number of continuous attributes, $d_V$: average node degree, CC: clustering coefficient, DIA: diameter}\label{table:net}
\resizebox{\columnwidth}{!}{%
\begin{tabular}{@{}rrrcp{0.25em}cccp{0.5cm}@{}} \toprule
 Networks & N & E & DS & CT & $d_V$ & CC & DIA \\ \midrule
 Facebook & 4,039 & 88,234 & 3 & 0 & 43.69 & 0.27 & 8 \\
 Enron & 36,692 & 183,831 & 0 & 7  & 10.02 & 0.72 & 13\\
 Patent & 2,738,012 & 13,963,839  & 4 & 3 & 10.20 & 0.09 & 22 \\
 Pokec & 1,138,314 & 14,975,771  & 2 & 0 & 13.16 & 0.054 & 11\\
 Philosopher & 1,564 & 26,761 & 7,969 & 0 &  17.11 & 0.37 & 7 \\
 Twitter & 81,306 & 1,768,149 & 33,569 & 0 & 21.74 & 0.064 & 7 \\
 Google+ & 107,614 & 13,673,453 & 690 &0 &  127.04 & 0.66 & 6 \\
\bottomrule
\end{tabular}
}
\end{table}

In this section, we discuss the real world datasets and generators to synthesize attributed network datasets.

We consider an assortment of five real-world datasets from varied domains: Facebook, Patent, Enron, Pokec and Wikipedia. The networks differ in size, and in key network parameters: degree distribution, diameter and clustering coefficient. See~\Cref{table:net} for a summary. The networks also differ in  attribute cardinality, attribute type (discrete vs. continuous attributes), data skew and assortivity (e.g. Patent category is most assortative with value 0.64). See ~\Cref{table:attr}.

The Facebook network~\cite{mcauley2012learning} is a friendship network. This network has discrete attributes of moderate cardinality and low assortativity (maximum assortativity is 0.34 for locale). 

The patent network~\cite{leskovec2005graphs} is the citation network of all patents granted by the US from 1963 till 1999. The attributes have high discrete cardinality for some of the attributes such as country of origin and continuous attributes like claims and citations have range over thousands. Most of the attributes are dis-assortative with exception of category and assignee type whose the assortativity values are 0.64 and 0.25 respectively. 

In the Enron network~\cite{leskovec2009community}, each node is an individual and edges represent communication between the corresponding individuals. The attributes vary greatly in range but have low assortativity values. 

Pokec \citep{takac2012data} is another social network from Slovakia. We use two discrete attributes: ``age'' and ``gender''. These attribute have low cardinality. The attribute ``gender'' is dis-assortatively mixed (-0.12) while ``age'' groups are homophilic (0.366). 

The Wikipedia network~\cite{ahn2010link} is an information network connecting philosopher pages in Wikipedia. Two philosophers share an attribute if they point to another non-philosopher page in Wikipedia. We treat a non-philosopher as an attribute if at least five philosopher pages cite it. The dataset is unique: the number of attributes per node is greater than the number of nodes; each attribute is boolean and asymmetric (i.e. one value is much more likely than the other). We perform our experiments on the largest component of the undirected versions of these networks. ~\Cref{table:attr} describes the attributes with their properties used in our experiments.

 \begin{table}[h!]
 \center
 \caption{Attribute statistics for three networks: Facebook, Patent and Enron. Type(TP): Continuous(CT) or Discrete(DS); CD : attribute cardinality; SW: Skew; CV: attribute coverage over the nodes in the network; AS: assortativity. The coverage of continous attribute is measured as their coverage of 10 log spaced bins. We don't show the attributes of Wikipedia due to its large size (7969). Due to very dense categorization in the original dataset, we use the ``subcat'' from the original dataset as our category and classes of the patent classes as subcategories.}\label{table:attr}
 \resizebox{\columnwidth}{!}{%
 \begin{tabular}{@{}r r r r r r@{}}\toprule
  Attribute & TP & CD & SW & CV & AS\\ \midrule
  \textit{Facebook} \\
  gender & DS & 2 & 0.06 & 97.92 & 0.09\\
  locale & DS & 10 & 0.69 & 98.56 & 0.34 \\
 education type & DS & 3 & 0.52 & 74.60 & 0.08 \\ \midrule

  \textit{Patent} \\
  category & DS & 36 & 0.06 & 100.0 &  0.64 \\
  sub-category & DS & 58 & 0.52 & 100.0 & -0.02 \\
  assignee type & DS & 7 & 0.40 & 100.0 & 0.25 \\
  country of origin & DS & 69 & 0.70 & 100.0 & 0.20\\
  citations made & CT& [0, 254] & 1.00 & 100.0 & -0.00\\
  citations received & CT& [0, 2142] & 0.31 & 100.0 & 0.06 \\
  claims made  & CT& [0, 263] & 1.00 & 100.0 & 0.04\\ \midrule

  \textit{Enron} \\
 AvgContentLength & CT & [0, 9296] & 1.00 & 100.0 & 0.02 \\
 AvgContentReplyCount & CT & [0, 1238] & 0.92 & 100.0 & 0.10 \\
 AvgNumberTo & CT & [0, 2653] & 0.95 & 100.0 & -0.03 \\
 AvgContentForwarding & CT & [0, 1004] & 0.85 & 100.0 & 0.10\\
 AvgNumberCc & CT & [0, 1827] & 0.99 & 100.0 & 0.04\\
 AvgRangeBetween2Mails & CT & [0, 10313] & 1.00 & 100.0 & -0.00\\ \midrule

\textit{Pokec} \\
 gender & DS & 2 & 1.3e-4 & 100.0 & -0.12\\
 age (group) & DS & 23 & 0.38 & 100.0 & 0.366\\

 \bottomrule  
 \end{tabular}
 }

 \end{table}

We decided against using Google+ and Twitter, cited in~\citet{yang2013community}, since these datasets have significant number of nodes with missing attribute values; missing values creates a confound since we don't know if the missing values are due to improper sampling of the original graph. Hence we've used real-world network datasets with the fewest missing attributes.

We now discuss the synthetic attributed-network generation. There are three elements to synthetic network generation: the network structure, the attributes and the relationship between attributes and network structure.

For the network generation, we use the Lancichinetti-Fortunato-Radicchi (LFR) \cite{lancichinetti2009benchmarks}  algorithm to generate artificial networks of size $N=1000$, with mixing coefficient $\mu=0.1$ that resemble real world networks. with strong community structure. Such networks are referred to as LFR($\mu$ =0.1), in the rest of the paper.

There are three essential data characteristics: skew, purity and assortivity. Assume that we have a single discrete attribute that takes on $k$ values; this discussion is easily extended to multiple discrete attributes and to continuous attributes. Now, in the discrete one dimensional case, all data points sharing the same distinct attribute value will be grouped together into one cluster,  $C_k$.

The data skew $s(\mathbb{C})$ ($=1 - H(\mathbb{C})/H_{max}$) of a set of clusters $\mathbb{C} = \{C_1, C_2, \dots, C_k \}$ is defined in terms of Shannon entropy $H(\mathbb{C})$ over the cluster size. We shall use three discrete skew values of (low $\approx 0$, medium $\approx 0.22$, high $\approx 0.52$) when generating the attributed network. The purity $p$ of the data refers to the separability of the data clusters; this is parameterized by the standard deviation of the continuous variable and is easily extended to discrete~\cite{pelechrinis2014matching}. We use two extreme purity values (low $\approx 0.2$ and high $\approx 10$) to synthesize the network. Assortativity ($a$) measures the degree to which nodes similar attributes are connected to each other that is significantly different from random matches~\cite{newman2003mixing}. We use two levels (low $\approx 0$, high $\approx 1$) of assortivity, where the low case corresponds to random assignments of attributes and the high case corresponds to $a \approx 1$. We note that negative assortivity is hard to achieve when the number of attribute values is large since cross attribute edges are similar to random matches.

Finally, we assign the synthesized attributes to the nodes in the network according to the specified assortative value through a label propagation algorithm, which is terminated when the target assortivity is achieved. We use the principle of \textit{swapping} and \textit{propagation} to map the attributes (content) onto the network. In the swapping technique, an extreme (high, low) assortative distribution over network is first generated using the community detection and approximate k-coloring problem respectively. Randomly swapping categorical attribute between pair of vertices causes assortativity to tend to zero which is stopped when the target assortativity is achieved. The propagation algorithm propagates the same category with a proportional high probability if the target assortativity is high and vice versa.

In this Section, we defined all the symbols used in the paper, and formally defined the task-driven sampling problem; we specifically focus on link-trace samplers in this paper. Then we discussed the real-world and synthetic datasets used in this paper. In the next Section, we discuss different network sampling methodologies.

\section{Sampling Attributed Networks}
\label{sec:How to Sample}

Let us denote the sample set of nodes collected from the network as $\mathbb{S}$. Frontier nodes, denoted as $N(\mathbb{S})$, are the set of nodes that have at least one neighbor in $\mathbb{S}$. We define $\Delta v$ to be the neighbors of node $v \in N(\mathbb{S})$ that do not belong to $\mathbb{S}$. Furthermore, we are primarily interested in acquiring a representative sample of \textit{node content} : attributes of a node unrelated to node structure; attributes such as location, gender, education level.

The set of sampling algorithms proposed in this paper fall under the description of link-trace sampling. In such sampling schemes, the next node $v$ selected for inclusion in the set $\mathbb{S}$ is a neighbor of at least one node in $\mathbb{S}$. In other words, $v \in N(\mathbb{S})$. We keep adding nodes until we have collected a target number of nodes $|\mathbb{S}|$; typically $|\mathbb{S}| \ll N|$, where $N$ is the number of nodes in the graph.

The mechanism of node addition (link trace vs. uniform sampling) has important implications on the sampled node content quality. Link trace sampling is essential when random access to individual nodes is unavailable. Social networks such as Facebook or Pinterest prevent random node access; notice that link-trace sampling techniques (e.g. BFS, Random Walk) are primarily used to crawl the World Wide Web. Furthermore, even if there is a unique numeric \texttt{id} associated with each individual (e.g. Twitter), the \texttt{id}s may not be sequential, requiring us to do costly rejection sampling~\cite{Gjoka2010}. The inability to randomly access each node in the graph (e.g. Facebook) implies that sampling uniformly at random for node content is infeasible on these networks. This is important because sampling uniformly at random is critical to creating a representative sample on which we can perform data mining tasks.

We can broadly categorize sampling algorithms as either attribute-agnostic, or attribute-aware. The key difference lies in if the sampling algorithm uses the content attributes (\textit{not related} to network structure) or not. We first discuss attribute-agnostic sampling, followed by a discussion of attribute-aware samplers in~\Cref{sub:Attribute Aware}.

\subsection{Attribute Agnostic}
\label{sub:Attribute Agnostic}

Attribute-agnostic algorithms, including breadth first search (BFS), Forest Fire~\cite{leskovec2006sampling}, Metropolis Hastings Random Walk (MHRW)~\cite{Hubler2008}, Random Walk  (RW) and Expansion Sampling (XS)~\cite{maiya2011benefits}, do not use the attribute (or content) of any node to construct $\mathbb{S}$. Well known sampling algorithms such as Forest Fire~\cite{leskovec2006sampling} ignore nodal content because they were explicitly designed to preserve graph structural properties in the sampled graph including degree distribution, diameter, and densification.

Now, we introduce well known attribute-agnostic sampling algorithms: snowball sampling; ForestFire; expansion sampling; Random Walk (RW) and MHRW. In snowball sampling, uses a small seed set of vertices ($\theta$) to start collection of data through Breadth First Search (BFS). Snowball sampling is computationally efficient, but biased towards high degree nodes~\cite{kurant2010bias} and is sensitive the selection of the seed nodes~\cite{maiya2011benefits}.~\citet{leskovec2006sampling} proposed ForestFire, which explores a subset of a node's neighbors according to a ``burning probability'' $p_f$; at $p_f=1$, ForestFire is identical with BFS.
At each iteration, the algorithm chooses a subset of the neighbors of the current node $v$ using a geometric distribution. While Forest Fire is superior to BFS, it suffers from a degree bias~\cite{kurant2010bias}.~\citet{maiya2011benefits} proposed expansion sampling (XS), motivated by expander graphs. XS adds nodes in a greedy manner in the direction of the largest unexplored region. That is, we a add a node $v^\ast$ to $\mathbb{S}$ when 
 \begin{equation}\label{eq:XS}
 v^\ast = \argmax_{v \in N(\mathbb{S})} |N(v) - (\mathbb{S} \cup N(\mathbb{S}) )|.
 \end{equation}

In other words XS finds that node $v^\ast$ that has the largest number of neighbors outside of the set $\mathbb{S} \cup N(\mathbb{S})$.~\citet{maiya2011benefits} suggest that XS relatively insensitive to seed set. While XS will rapidly discover homogeneous communities, it does less well over disassortative networks since the sampling algorithm is attribute agnostic. Re-weighted Random Walk sampling (RW) is a variant of the classic Random Walk algorithm, re-weighted to provide a better estimate of the content distribution. The re-weighting is necessary since the random walk algorithm has a high degree node bias. Notice that the stationary probability $\pi_v$ of visiting a node $v$ is proportional to the node degree. Hence the re-weighting discounts the label associated with the node $v$ of degree $d_v$, by its degree; attribute probabilities are estimated through the Hansen-Hurwitz estimator \cite{hansen1943theory} to develop an unbiased estimate of the content. Assuming an attribute $A$ can take values ($A_1$, $A_2$, ... $A_r$) with the corresponding groups; $\cup_1^r A_i = V$, the unbiased probability distribution ($\widetilde{p}$) estimate of any discrete attribute $A$ from a RW sampled collection is :
 \begin{align}
 \widetilde{p}(A_i) = \frac{ \sum_{u \in A_i} 1/d_u }{ \sum_{u \in V} 1/d_u }
 \end{align}
 Similarly, we can use kernel density estimators for continuous attributes.

Metropolis-Hasting random walk sampling (MHRW) has an important asymptotic property: the stationary distribution is uniform over all the nodes. Thus in principle, MHRW is equivalent to uniform random sampling of the graph. MHRW achieves the uniform stationary distribution by altering the transition probabilities between pairs of nodes. One important concern: less than ideal finite sample behavior.  Poor finite sample behavior is observable on graphs with strong community structure, causing the MHRW to get stuck in a local community. MHRW typically requires sample sizes of $O(N)$, where $N$ is the number of nodes in the graph, to achieve the stationary distribution. The finite sample performance of MHRW becomes problematic for typical sample sizes for internet scale graphs (e.g. Facebook has over a billion users) that are an  order of magnitude smaller than $N$ (i.e. $|\mathbb{S}| \approx 0.05 \times N$; 5\% of $N$).

\subsection{Attribute Aware}
\label{sub:Attribute Aware}
Attribute-aware samplers use node attributes (content) to determine the sample set $\mathbb{S}$. These samplers determine the next node $v$ to be added to the current sample set $\mathbb{S}$, by checking the content of the node against the content of the nodes in the current sample. At any point in time, we have the set $\mathbb{S}$, comprising nodes in the current sample set, $N(\mathbb{S})$, the set of all frontier nodes who have at least one neighbor in $\mathbb{S}$. At each step, we shall  add to $\mathbb{S}$,  one optimal node $v \in N(\mathbb{S})$. We shall assume that for each $v \in N(\mathbb{S})$, we shall have access to the content of the neighbors of $v$. This is similar in spirit to Expansion Sampling (XS) proposed by~\citet{maiya2011benefits}. We call the set of neighbors of $v$, that do not belong to $\mathbb{S}$, the candidate set for node $v$ and shall designate it with $\Delta v$.

The rest of this section organized as follows. In the next section, we introduce the idea of surprise, grounded in Information Theory and develop  algorithms that incorporate surprise to sample network content. Then, in~\Cref{subs:Extremal Point Sampling}, we introduce the idea that extremal points---ones that are far away from all the points in the current sample---are the most informative.
\subsubsection{Surprise Based Sampling}
\label{subs:Surprise Based Sampling}
Surprise based samplers compute the extent to which the distribution of attribute values in the candidate set $v \cup \Delta v$ is predicted by the set $\mathbb{S}$.

Balanced sampling (BAL) is the simplest surprise-based sampler that adds one node from the frontier at a time; the attributes of the selected node have low probability of occurrence in the sample $\mathbb{S}$.

In Information Expansion Sampling (IXS), surprise $I_{\Delta v}$ of a candidate set $v \cup \Delta v$ (with respect to $\mathbb{S}$) is computed as follows:
\begin{align}
I_{\Delta v} =& \frac{- \ln P(\Delta v | \mathbb{S})}{|\Delta v|} \nonumber\\
=& -\sum_{i=1}^r p_{\Delta v}(i) \ln p_{\mathbb{S}}(i) \label{eq:surprise}
\end{align}
where, $r$ is the number of distinct attribute values, $p_{\Delta v}(i)$ is the probability of attribute value $i$ in the candidate set $v \cup \Delta v$, and $p_{\mathbb{S}}(i)$ is the probability of the attribute value $i$ in the sample set $\mathbb{S}$. IXS  expands to rapidly discover unseen attribute values. Notice that unseen attribute values  (i.e. $p_{\mathbb{S}}(i)=0$) in $\mathbb{S}$ will cause~\Cref{eq:surprise} to diverge; IXS reduces to BAL when $\Delta v = \phi$. Selection of nodes from the neighborhood set $N(\mathbb{S})$ in a manner that maximizes surprise~\Cref{eq:surprise} results in interesting sampler behavior over time. We identify the behavior through two lemmas.

\begin{lemma}\label{lemma:coverage}
    The IXS sampler prefers nodes with unseen attribute values to nodes with attribute values present in $\mathbb{S}$.
\end{lemma}
\begin{proof}
    The surprise of a candidate set $v \cup \Delta v$ diverges, when the set contains a node with an attribute value not present in the current sample $\mathbb{S}$. The divergence can occur either because the node $v$ has an attribute value absent in $\mathbb{S}$, or that $\Delta v$ has at least one node with an unseen attribute value. In the former case, we add node $v$ to the sample immediately, while in the latter case, we add the node in $\Delta v$ with the unseen attribute value in the next step.
\end{proof}


%


Assume that we have an infinite $d$-regular random graph, where each node of the graph takes on a value from a categorical attribute $X$. Then, $e_{ij}$ is the probability that node with attribute value $i$ has a neighbor with attribute value $j$, with $\sum_j e_{ij}=1$. The set ${e_{ij}}$ defines the assortivity matrix $E$ for attribute $X$. The assortivity matrix is tied to the distribution $p$ of attribute values of $X$ over the entire graph as follows: $E^t p = p$. That is, $p$ is the right eigenvector of $E^t$ with eigenvalue 1. Now, we prove a Lemma on how the  assortivity $E$ creates a bias in information expansion sampling.

\begin{lemma}\label{lemma:bxs}
    The distribution of attribute values in the sample set $\mathbb{S}$, tends to the population distribution of the attribute, under Information Expansion Sampling over a $d$-regular, random, infinite graph.
\end{lemma}
\begin{proof}
We shall prove the result for the case of an attribute that takes on two values \{1,2\}. Let us assume that the two attribute values $\{1,2\}$ occur in the set $\mathbb{S}$ with probabilities $p$ and $1-p$ with $p < 1/2$.  Further assume that in the sample candidate set $v \cup \Delta v$, attribute value $1$ occurs with probability $x$ and attribute value 2 occurs with probability $1-x$. Thus the attribute value with lower probability occurs in the candidate set $v \cup \Delta v$ with probability $x$. The surprise $I_{\Delta v}$ associated with the set $v \cup \Delta v$ is:
\begin{align}
I_{\Delta v} =&  - x \ln p - (1-x) \ln (1-p) \nonumber \\
=& -\ln (1-p) + x \ln \frac{1-p}{p}.
\end{align}
Since $I_{\Delta v}$ is linear in $x$ and since $p < 1/2$, IXS will pick the node $v^\ast \in N(\mathbb{S})$ with the largest value of $x$ to maximize $I_{\Delta v}$. In other words, IXS will pick the node $v^\ast$ with the largest fraction of least probable attribute.

If $p<x$, then the entropy of the updated sample set $\mathbb{S}\cup v$ increases as the entropy $H(p)$ of the sample $\mathbb{S}$ is concave in $p$. The limit of $p$ is simply $\mathbb{E}(x)$, the expectation of the fraction of the candidate set that is of attribute value 1. In other words, $p \leq \mathbb{E} (x)$.

The expected value of $x$, $\mathbb{E} (x)$, depends on the assortivity matrix $E$ as well as the probabilities $p_1$ and $p_2$, the probabilities that a random node is of type 1 or type 2. We compute $\mathbb{E} (x)$ as follows:
\begin{align}
\mathbb{E} (x) &= p_1 \frac{(d-1)e_{11} + 1}{d} + p_2\frac{(d-1)e_{21}}{d} \label{eq:expectation}\\
&= \frac{(d-1)p_1+ p_1}{d} \label{eq:eigenvector}\\
&= p_1.
\end{align}
\Cref{eq:expectation} shows that $\mathbb{E}(x)$ is the product of the prior probabilities of node $v$ having attribute 1 or 2, times the expected fraction number of neighbors being of type 1 given that $v$ either has attribute 1 or 2. The expected fractions depend on the assortivity values and $d$ the degree each node.~\Cref{eq:eigenvector} follows as $[p_1, p_2 ]$ is a left eigenvector of the assortivity matrix $E$.  The result that $p=\mathbb{E}(x)$ follows due to the concavity of $H(p)$.
\end{proof}

To summarize, IXS begins with a bias to rapidly cover the range of attribute values, and in the limit, ensures that the distribution of attribute values in $\mathbb{S}$ tends to the attribute value distribution in the population. In other words, the initial behavior is akin to stratified sampling and the limiting behavior of IXS is akin to uniform sampling of the node content.

Thus far, we've discussed surprise for categorical variables. We can extend the notion of surprise to continuous variables by simply discretizing the continuous variables and then computing the surprise using~\Cref{eq:surprise} with the discretized values. Next, we discuss a different approach to sampling continuous content.

\subsubsection{Extremal Point Sampling}
\label{subs:Extremal Point Sampling}

A node that is at a large distance, in terms of its features, from the all the current nodes would be surprising. We use this idea, which we term extremal point sampling (ExP), to identify surprising nodes for nodes with continuous features. In ExP, we rank the candidate nodes in terms of their  average distance to all the nodes in the sample set $\mathbb{S}$. The node with the highest rank is then added to the sample set $\mathbb{S}$. While we could use many different distance measures, we choose to use the standard Euclidean distance. Mahanalobis distance with its covariance correction would be the ideal Euclidean distance choice, but is not used due to the difficulty in developing a stable estimate for the covariance matrix with a small sample. Information expansion samplers that use both continuous and discrete variables are termed as Hybrid IXS  (H-IXS).

 \subsubsection{Surprise based MHRW}
 \label{subs:Surprise based MHRW}

 Could we make MHRW attribute aware? One possibility is to couple the surprise for each node $v \in N(i)$, where $N(i)$ is the neighborhood of $i$, the node where the MHRW sampler is at present. We could define the probability $\hat{p}_{i,v}$ of jumping from node $i$ to node $v$ as:
 \begin{equation}\label{eq:probabilistic MHRW}
 \hat{p}_{i,v} \propto p_{i,v} I_{\Delta v}
 \end{equation}

 where, $I_{\Delta v}$ is the surprise with respect to the sample set $\mathbb{S}$ and $p_{i,v}$ is the probability of transitioning to $v$ from $i$ in the original MHRW sampler.

 This approach has intuitive appeal since it appears to combine the best ideas from attribute-agnostic samplers with that of surprise based attribute-aware samplers; in addition unlike IXS or H-IXS, it is not a deterministic algorithm. The challenge is that~\Cref{eq:probabilistic MHRW} changes the stationary distribution of the sampler---we are no longer guaranteed uniform stationary distribution over the graph nodes. Such content-aware MHRW algorithms are also harder to analyze since the process is no longer first order Markov. Regardless, the idea that one could combine attribute-aware and attribute-agnostic samplers has obvious appeal, and we shall consider this idea in more detail in~\Cref{sub:Pareto Optimal}. We propose a simplistic combination sampler from IXS and MHRW that chooses non-deterministically with equal probability to sample from either of the strategies. We call this combination sampler as IXS and MHRW or I\&M.

 \label{sub:Sampling with Side Information}
 All the algorithms discussed thus far assume no prior knowledge of the structural  characteristics of the network or anything about the distribution of the content. However, often, we may have a rough idea of either the properties of the network, say the skewness of the degree distribution or of the attribute values (e.g. most of the Twitter users are from the U.S.).  How should we incorporate this side information into the sampling process?

 We have explored this idea with respect to sampling content properties of the network. For example, assume that we have access via an oracle, to the underlying attribute distribution $\bm{p}$ over the entire network. Then, one could simply use the earlier surprise based criteria to add additional nodes to the sample $\mathbb{S}$, except that instead of using $P_{\mathbb{S}}(i)$ in~\Cref{eq:surprise}, we use $p_i$. Notice that $\bm{p}$ is a constant while $P_{\mathbb{S}}$ is variable. This change will ensure that samples collected in $\mathbb{S}$ will have an attribute distribution that matches $\bm{p}$.

 When the prior $\bm{p}$ is unavailable, one could proceed as follows. We first MHRW till the sample statistic (say the distribution mean) converges via Gilman Ruben or Gweeke statistic~\cite{Gjoka2010} and then estimate $\hat{\bm{p}}$, the sample attribute value distribution. Then, we proceed as earlier and use $\hat{\bm{p}}$ in the surprise calculation.

 Another approach is to incorporate knowledge of the underlying content clusters, or content classes, which may be known (e.g. if all residents in a U.S. state form a class, then there are 50 classes). Assume then, that we know the number of content clusters $k$. We can combine the IXS (for categorical content) and the ExP samplers (for the continuous feature vectors) as follows. From the sampled set $\mathbb{S}$, we construct $k$ content clusters with different centers using the continuous content. Then, we assign each node in $\mathbb{S}$ and $\Delta v$ to the nearest cluster center. Now, we can compute surprise as earlier based on the cluster \texttt{id} distribution of $\Delta v$ in conjunction with the surprise of $\Delta v$ with respect to the distribution of categorical attributes in $\mathbb{S}$.

 On the other hand if the content distribution of the original network is known from practice or approximation \cite{hubler2008metropolis}, we leverage the variable neighbourhood search (VNS) approach to select nodes sequentially that preserve the known prior content distribution in the hope of gaining a better representative sample.  In other words we select the node $v\ast$ in $N(S)$ which minimizes KS statistic of the sample and original distribution.

 \begin{align}
 v\ast = argmin_{v \in N(S)} D(A(V_S), A(V))
 \end{align}

 \subsection{Pareto-Optimal Sampling}
 \label{sub:Pareto Optimal}

 It would be ideal to develop a sampler that could preserve the properties of the network content as well as structural properties of the network. We do this via a Pareto-optimal sampler that combines MHRW based sampling with surprise based sampling. Assume that we have a sample set $\mathbb{S}$. Then,  $\forall v  \in N(\mathbb{S})$, we can compute two numbers: the probability of reaching $v$ from $\mathbb{S}$ via MHRW and $I_{\Delta v}$ the surprise due to node $v$. Thus we can compute the Pareto-optimal frontier using all $v  \in N(\mathbb{S})$ and choose the node from the frontier that best suits our bias (equal weight to structural properties and to content). We provide results of a pareto-combination of IXS ($I_{\Delta v}$) and XS ($|\Delta v|$) called as pIX (pareto-IXS-XS) and another pareto combination of IXS ($I_{\Delta v}$) and MHRW (under independece assumption of new transition probabilities ($P^{\ast}$) called as pIM (pareto-IXS-MHRW). 
 \begin{align}
 P_{u,v}^{\ast} = \frac{ \sum_{u \in S} min(1/d_v , 1/d_u) }{ \sum_{w \in N(S)}  \sum_{u \in S} min(1/d_w , 1/d_u) }
 \end{align}

In this section we discussed attribute-agnostic and attribute-aware sampling schemes. The main idea behind attribute-aware sampling algorithms is to preserve node content properties including content distribution in the sample. We introduced the idea of surprise, grounded in Information Theory, as a metric to develop sampling schemes. Next, we evaluate these sampling schemes on characterizing the node content.


\section{Data Characterization}
\label{sec:DataCharacterization}

In this section, we discuss how samplers preserve the statistical characteristics of attributed graphs: properties related to the network structure; distributions of attributes and joint content-network relationships. We conclude this section by presenting experimental results comparing different samplers.

\subsection{Properties}
\label{sub:Network Characteristics}

We study three properties of an attributed network: network structure, content structure and network-content relationship.

\textit{Network Properties: }We use three properties widely used to characterize network datasets~\citep{maiya2011benefits,leskovec2006sampling}---\emph{degree distribution, clustering coefficient, diameter}. We will evaluate samplers based on their ability to preserve these three network properties in the sampled subgraph.
that
The degree distribution is simply the probability distribution $P(k)$ of finding a node with degree $k$ in the network. The clustering coefficient distribution is the distribution of clustering coefficients over node degree. The local clustering coefficient of a node $v$ is defined as: $C_v= 2e_v/d_v \times (d_v-1)$, where $e_v$ is the number of edges amongst the neighbors of node $v$, and $d_v$ is the degree of node $v$.


\textit{Content Characteristics: }We would like samplers to preserve essential aspects of the node content, including attribute value distribution and attribute coverage. By the phrase ``node content,'' we refer to the attributes such as ``gender=female,'' ``ethnicity=asian''; we are using the word ``content'' to refer to all nodal attributes that are \textit{not} derived from structural properties of the graph, such as degree and clustering coefficient. We use the familiar Kolmogorov-Smirov (KS) statistic to compute the distance between the sample attribute value distribution and the underlying ground-truth attribute value distribution. \textit{Content coverage} is another key content characteristic. We define content coverage as the ratio of the number of unique attribute values in the sample to the cardinality of the corresponding attribute in the underlying content. We use logarithmic binning for continuous attributes. Besides distribution and coverage, content attributes exhibit structure in the form of clusters; we discuss this in \Cref{sec:Discovering data clusters,sec:Classification}.


\textit{Joint Network-Content Relationships} : Network structure and node content are often correlated; this is termed as homophily. For example, the correlation can arise due to homophily \citep{McPherson2001} when friendships form when like minded individuals seek out each other. Thus it is important to preserve the correlation between network and content. We will use assortativity \citep{newman2003mixing} a widely used metric to measure this correlation. Besides assortativity, we define and discuss more specific measures of homophily at nodal levels such as Ego-relation and Star-relation \citep{Kumar2016}.

However, assortativity being a global measure fails to capture attribute mixing at micro levels. We therefore propose two new local measures-- Star-relation ($S_v$) and Ego-relation ($E_v$) for every attribute. Star-relation is the defined as the agreement (or correlation) of content attribute values \textit{between} a node $v$ and her friends. On the other hand, Ego-relation ($E_v$) is defined as the agreement of content \textit{among} the node $v$ and her friends. Thus,

\begin{align*}
	S_v =& \frac{|\{u \in N_v \mid A(u) = A(v)\}|}{d_v}, \\
	E_v =& \frac{2|\{(u, w) \in  n_v \cup v \mid A(u) = A(w)\}|}{d_v(d_v-1)}.
\end{align*}

The equations says that Star-relation of a node $v$ is the fraction of $v$'s neighbor who have same attribute value as $v$. Notice that the degree $d_v$ of the node $v$ is simply $|N_v|$. Similarly, the Ego-relation of a node $v$ is the fraction of node pairs in neighborhood of $v$ including $v$ that have the same attribute value. Thus $S_v$ captures the degree to which a node $v$ agrees with her friends, while $E_v$ captures the degree to which a group agrees. In this paper, we study assortativity, Star-relation and Ego-relation for every attribute independently to understand the different network-content relationship.

\subsection{Experimental setup}
%


\textit{Dataset description: } We now present the dataset used along-with our guiding principle for choosing these networks and their corresponding content attributes. While many attributed networks are available, we decided to work with only those networks that are not sparse---in other words, most nodes have values for attributes of interest. We set the sparsity threshold to 75\%. We removed nodes with missing values from our datasets in the pre-processing step. Hence, we chose to work with network datasets from Facebook, US Patent, Enron, Wikipedia and Pokec. Furthermore, we considered attributed-network datasets such as Google+, Twitter and Microsoft Academic Search dataset \citep{leskovec2009community,sinha2015overview} which were not chosen due to concerns over sparsity. We recognize that future sampler design must address the issue of sparsity and noise. We picked a set of attributes to work with, such that over these attributes we had a wide variation in range, cardinality, purity, skew and assortativity. We choose attributes from that are not dependent on network structure such as ``gender'' of a person and ``category'' of a patent which are independent of network. Additionally for attribute choice, we choose contrasting attributes to cover a wide range of content characteristics such as range/cardinality, purity, skew and assortativity.

\textit{Evaluation: } We evaluate the performance of a sampler as the mean performance for a specific characteristic over a range of sample sizes. As an example, for a fixed sample size, we use the K-S statistic to measure the difference in attribute value distribution between the original graph and the sample. We compute the expected value, by running the sampling operation using a new seed node hundred times. Then we compute the mean of these expected values over the different target sample sizes to determine the average performance of the sampler.

Thus, let $D$ be the measure used, with  $\alpha(l)$ being the number of nodes in the sampled graph $G'$, such that the sampled graph size is $l$\% of the original graph $G$. Thus at any target sample size $\alpha(l)$, the expected performance is $\mathbb{E}[D(G, G'; \alpha(l))]$, where the expectation operator $\mathbb{E}$ is over different samples $G'$ each of size $\alpha(l)$. Thus the mean performance with respect to a characteristic is:
$$\bar{D}(G, G') = \frac{1}{Q} \sum_{l=1}^{Q} \mathbb{E}(D(G, G'; \alpha(l))) $$
The mean performance values $\bar{D}$ can be interpreted as area under the curve of performance $D$ against sample size. We use $Q = 10$.

\subsection{Experimental results}

\begin{table}[b]
\center
\caption{ The content distribution preservation goal is to have the least possible KS statistic (lower is better) averaged over all attributes between original and sampled content distribution. Consistent with theory, UNI is the best sampler to preserve attribute distributions.}
\label{table:content}
\resizebox{\columnwidth}{!}{%
\begin{tabular}{@{}cccccc@{}}
 & \multicolumn{1}{c}{\textit{Facebook}} & \multicolumn{1}{c}{\textit{Patent}}& \multicolumn{1}{c}{\textit{Enron}}& \multicolumn{1}{c}{\textit{Pokec}} & \multicolumn{1}{c}{\textit{Wikipedia}} \\ \toprule
BFS&0.109&0.061&0.623&0.155&0.541\\
RW&0.103&0.066&0.687&0.109&0.472\\
MHRW&0.129&0.068&0.627&0.048&0.514\\
FF&0.109&0.129&0.685&0.102&0.475\\
XS&0.047&0.351&0.628&0.246&0.326\\ \midrule
UNI&\textbf{0.038}&\textbf{0.016}&\textbf{0.412}&\textbf{0.004}&0.542\\ \midrule
ExP&N/A&0.148&0.438&N/A&N/A\\
BAL&0.290&0.219&N/A&0.308&\textbf{0.135}\\
IXS&0.181&0.138&N/A&0.283&0.269\\ \bottomrule

\end{tabular}
}
\end{table}

In this section, we present our sampling results for four different characteristics of attributed network: (1) content distribution, (2) content coverage, (3) network structure including degree, clustering coefficient and path length and (4) content-network dependence including assortativity. Results for the four characteristics are presented in the \Cref{table:content,table:coverage,table:network,table:joint-net-data}.




We now describe the general characteristics common to all results followed by detailed interpretation of each result table. Notice that the results in \Cref{table:content,table:coverage,table:network,table:joint-net-data} have segmented rows and some missing values (N/A). We organize the samplers into three groups (visually segmented in the tables): attribute-agnostic link-trace samplers such as BFS, RW, MHRW, FF and XS; the baseline sampler, UNI and the proposed surprise-driven attribute-aware samplers like ExP, BAL and IXS. Some entries are listed as Non applicable (N/A). This can happen if we use a sampler that relies on continuous attributes (ExP), but the network has only discrete attributes (Facebook). Similarly, in \Cref{table:content}, IXS and BAL cannot be performed over the continuous attributes in Enron without the knowledge of the range of continuous attribute.


\begin{table}[b]
\center
\caption{ The content coverage goal is to have maximum possible mean coverage of content (all attributes). IXS is the best sampler due to its biasness towards new attribute values.}
\label{table:coverage}
\resizebox{\columnwidth}{!}{%
\begin{tabular}{@{}cccccc@{}}
 & \multicolumn{1}{c}{\textit{Facebook}} & \multicolumn{1}{c}{\textit{Patent}}& \multicolumn{1}{c}{\textit{Enron}}& \multicolumn{1}{c}{\textit{Pokec}} & \multicolumn{1}{c}{\textit{Wikipedia}} \\ \toprule
BFS&0.752&0.892&0.167&0.978&0.452\\ 
RW&0.752&0.899&\textbf{0.192}&0.986&0.482\\ 
MHRW&0.648&0.896&0.161&0.989&0.443\\ 
FF&0.754&0.910&0.188&0.984&0.513\\ 
XS&0.825&0.868&0.154&0.962&0.660\\ \midrule
UNI&0.826&0.882&0.082&0.995&0.473\\ \midrule
ExP&N/A&0.924&0.053&N/A&N/A\\ 
BAL&0.784&0.942&N/A&\textbf{1.000}&\textbf{0.843}\\ 
NXS&\textbf{0.839}&\textbf{0.960}&N/A&\textbf{1.000}&0.722\\  \bottomrule

\end{tabular}
}
\end{table}

In the \textit{first} data characterization task, we compute the KS statistic between attribute distributions of the nodes in sampled graph and the nodes in original graph; we present the results averaged over all the attributes in the dataset. \Cref{table:content} depicts the result for content distribution. Observe that UNI, where nodes are selected uniformly at random, is the best sampler for preserving content distribution. This is because UNI creates an unbiased estimate of content distribution making it ideal for all attributes. In contrast, attribute-agnostic samplers like BFS are influenced by homophily in the network, and the proposed attribute-aware samplers are intrinsically biased towards \textit{new} attribute values. Surprisingly, even though MHRW and re-weighted RW have a uniform stationary distribution, it shows poor finite sample performance~\citep{Gjoka2010}.

For the \textit{second} data characterization task, we compute coverage of content (attribute) in sampled graph averaged over all attributes. \Cref{table:coverage} shows the efficiency of samplers at exploring different attribute values. Attribute-aware samplers like IXS and BAL, being biased towards rare attribute values (Lemma \ref{lemma:coverage}), perform much better for real-world attributes. This is due to the fact that most of the real-world attributes are highly skewed. The surprisingly improved performance of BAL over IXS in Wikipedia shall be discussed in detail in \Cref{sub:Experimental results clustering}. Due to homophily of content in real-world networks, attribute-agnostic samplers like MHRW and BFS perform poorly at content coverage. Note, ExP that covers extreme continuous values fails to cover mid-logarithmic bins in Enron.

\begin{table}[h]
\center
\caption{ The network-structure preservation goal is to obtain the least possible KS statistic of degree, clustering coefficient and path length distributions. Attribute-agnostic samplers such as XS, FF and RW are tuned to preserve network characteristics, perform better than attribute-aware samplers. }\label{table:network}
\resizebox{\columnwidth}{!}{%
\begin{tabular}{@{}c ccc ccc c@{}}
 & \multicolumn{3}{c}{\textit{Patent}} & \multicolumn{3}{c}{\textit{Enron}} \\ \toprule
Samplers& Degree& CC& Path& Degree& CC& Path\\ \toprule
BFS&0.074&0.075&0.688&0.342&0.295&0.679\\
RW&0.080&0.080&0.442&0.391&0.324&0.423\\
MHRW&0.071&0.072&0.492&0.315&0.308&\textbf{0.178}\\
FF&0.185&0.188&0.916&0.388&0.335&0.387\\
XS&0.114&\textbf{0.040}&0.554&\textbf{0.151}&\textbf{0.122}&0.251\\ \midrule
UNI&0.863&0.673&0.988&0.624&0.518&0.894\\ \midrule
ExP&\textbf{0.051}&0.044&\textbf{0.217}&0.367&0.303&0.610\\
BAL&0.088&0.088&0.534&N/A&N/A&N/A\\
IXS&0.100&0.100&0.479&N/A&N/A&N/A\\ \bottomrule
\end{tabular}%
}\end{table}

For the \textit{third} data characterization task, we compute the similarity in distribution of network features : degree, clustering and path length. We report the mean KS statistic between the sample and underlying network feature distributions in \Cref{table:network}. We observe similar behavior among all datasets, but due to space constraints, we show results from only Patent and Enron datasets. The network characteristic for the remaining datasets--Facebook, Pokec and Wikipedia--is shown in \Cref{table:remaining}.~\citet{leskovec2008planetary} suggest a fast, approximate path-length distribution computation by randomly selecting 1000 nodes to construct shortest path from the selected nodes to all other nodes in the network. As expected, attribute-agnostic samplers such as MHRW and XS outperform attribute-aware samplers. This is  because FF, XS, RW and MHRW are link-trace samplers designed to preserve network structure characteristics. UNI (sampling nodes uniformly at random) is the worst sampler due to low edge density in real-world networks; low edge density causes UNI sampled graphs to have a large number of disconnected components.

\begin{table}[!h]
\center
\caption{ The content-network relationship (assortativity) preservation goal is to have least possible difference in assortativity between the sampled and original attributed networks over all attributes. By expectation, random edge sampling would be the ideal. Among the samplers discussed, there is no single dominant sampler strategy for this task.}\label{table:joint-net-data}
\resizebox{\columnwidth}{!}{%
\begin{tabular}{@{}cccccc@{}}
 & \multicolumn{1}{c}{\textit{Facebook}} & \multicolumn{1}{c}{\textit{Patent}}& \multicolumn{1}{c}{\textit{Enron}}& \multicolumn{1}{c}{\textit{Pokec}} & \multicolumn{1}{c}{\textit{Wikipedia}} \\ \toprule
BFS&0.088&0.088&0.024&0.076&0.597\\
RW&0.091&0.091&0.013&0.025&0.597\\
MHRW&0.094&0.094&0.023&0.025&0.597\\
FF&0.081&0.081&\textbf{0.012}&0.020&0.597\\
XS&0.122&0.122&0.013&0.213&0.597\\ \midrule
UNI&\textbf{0.074}&0.074&0.044&\textbf{0.007}&0.689\\ \midrule
ExP&N/A&0.027&0.033&N/A&N/A\\
BAL&0.144&0.005&N/A&0.076&0.597\\
IXS&0.141&\textbf{0.003}&N/A&0.030&0.598\\ \bottomrule

\end{tabular}
}
\end{table}

In the \textit{fourth} data characterization task, we compute the absolute difference in assortivity between the sampled attributed graph and the original attributed graph. The absolute difference between assortativity values are shown in \Cref{table:joint-net-data}. Observe that there is no sampler that distinctly outperforms others at preserving assortativity. It remains an open problem to design an efficient sampler that can preserve network-content relationship over a wide range of networks.

In sum, we saw in this section that UNI performs best for content distribution sampling, attribute-aware samplers are efficient at exploring new content values, attribute-agnostic samplers preserve the network structure and there is no statistically-dominant strategy for preserving content-network relationship.




\section{Discovering Content Clusters}
\label{sec:Discovering data clusters}

In this section, we present experimental results that show the effects of different types of link-trace samplers on content clustering. First, we discuss the experimental setup including validation metrics. Finally in \Cref{sub:Experimental results clustering} we present our experimental results.

Clustering is a statistical technique to organize objects into groups. Consider a collection of $n$ objects \{$x_i$\}. The goal of clustering is to partition the collection into $k$ groups such that objects within each group are more similar to each other than with objects in other groups. In particular, we assume that the content attributes of the nodes in the network can be partitioned into $k$ groups. Thus the goal is to develop samplers that preserves the $k$ groups of the original data in the sample. In this paper, we shall consider the case of non-overlapping clusters. Furthermore, we are specifically interested in the use of the collection \{$x_i$\} referred as the content attributes of the nodes.

%

\subsection{Experimental setup}
\label{sub:Experimental setup1}

\begin{figure*}[!htb]
  \centering
\includegraphics[width=\textwidth, trim={9cm 0.3cm 8.65cm 1cm},clip]{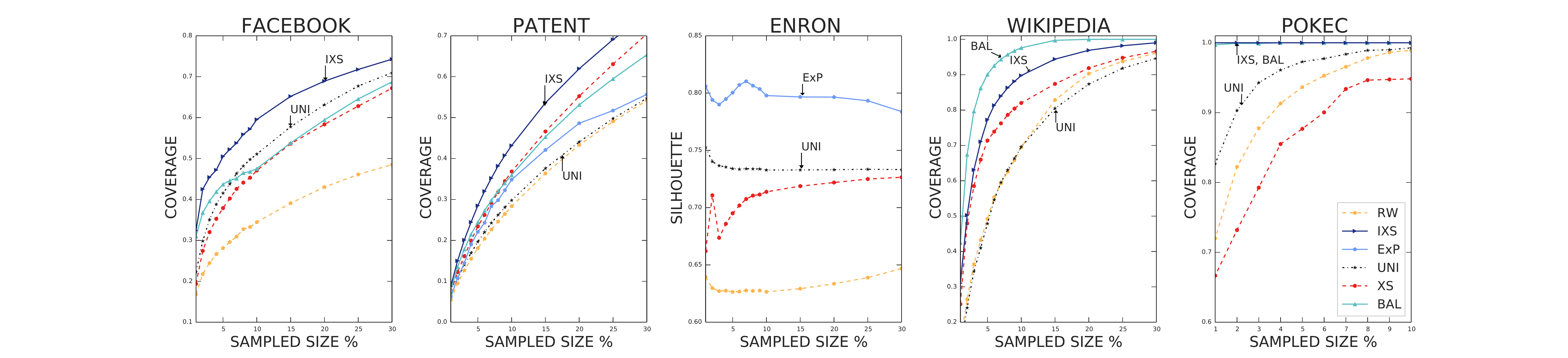}
\caption{Clustering performance on five real-world datasets---Facebook, Patent, Enron, Pokec and Wikipedia---shows attribute-aware samplers outperforming attribute-agnostic samplers as well as uniform sampling at nearly all sampling points. Due to the skewed cluster-size distribution, UNI is unable to discover small sized clusters.}
\label{fig:real_clusters}
\end{figure*}

In this section we discuss the datasets used in the experiments as well as the evaluation methodologies specific to clustering.

\textit{Dataset description: }We test on real-world and synthetic datasets. We study the impact of sampling on clustering performance on all five real-world attributed graphs: Facebook, Patent, Enron, Pokec and Wikipedia. Synthetic datasets are valuable because not only can we control the ground-truth clusters, but we can also vary the characteristics such as purity, skew and assortativity and observe their effects on the clustering results. We employ synthetic network generation model discussed in \Cref{sub:Datasets} to generate attributed networks from several network generators including LFR, Watts-Strogatz and Barabasi models with varying content characteristics of purity, skew and assortativity. Additionally when synthesizing the data, we assume that each object belongs to only one cluster. Furthermore, the model generates three attributes---two continuous and one discrete attribute---for each node or object. The attribute values are influenced by parameters of purity, skew and assortativity.

\textit{Evaluation: } Clusters from real-world and synthetic clustering are evaluated differently. In absence of ground-truth clusters, we use two different metrics for real-world datasets: cluster coverage and silhouette coefficient. Cluster coverage is defined over discrete attributes in Facebook, Patent and Pokec as the ratio of the number of unique attribute combinations (content cluster) in the sample to the maximum number of attribute combinations. Due to the very high attribute size of Wikipedia, we define the coverage of attributes as the ratio of number of attributes observed in sample to the total number of attributes in the underlying network. Silhouette coefficient captures the compactness of clusters and is used here to evaluate clustering quality in the Enron dataset, which comprises continuous attributes. For synthetic networks, since we know the ground truth label assignments, we use Normalized Mutual Information (NMI) to evaluate the clusters.

\subsection{Experimental results}
\label{sub:Experimental results clustering}

Now, we present experimental results for real-world and synthetic networks. We begin with a discussion of results for real-world networks.

%
%



We present the clustering performance for five real-world datasets: Facebook, Patent, Enron, Wikipedia and Pokec. We observe very similar behavior of attribute-agnostic samplers, FF, MHRW, BFS and RW; we only plot RW as the representative of random walk variants. \Cref{fig:real_clusters} shows attribute-aware samplers like IXS significantly outperform attribute-agnostic samplers such as RW and XS. The results show expected behavior. Being biased towards unseen attributes (Lemma~\ref{lemma:coverage}), IXS discovers content clusters (unique attribute combinations) much faster than RW samplers. RIXS and ExP does remarkably better than UNI across all attributed networks. At low sampling percentages (1, 2, 3\%), IXS's performance is $>$144\% better than UNI's performance in Pokec and Wikipedia networks, while it is $>$40\% better in Facebook and Patent. Likewise, ExP is better than UNI by a margin of 10\% in Enron network.

Wikipedia is a highly unusual dataset. It has a very large number of attributes (7,969) compared to just 1,564 philosopher nodes. On average, there are five unique attributes per node. Furthermore, the attributes in Wikipedia are highly asymmetric binary variables, i.e. probability of a philosopher page having some attribute like ``feminist = True'' or ``China = True'' is very small (\~1\%). As a consequence, the surprise created by almost every philosopher $v$ in the frontier set diverges. The divergence is amplified when the $\Delta v$ is large, thereby subduing the surprise (divergence) from the node $v$. As a result BAL, the information sampler that has smallest $\Delta v = \phi$ has the best performance (fourth subfigure of~\Cref{fig:real_clusters}).


The information surprise defined by IXS can easily be extended to handle asymmetric binary attributes. Thus, the information surprise can be conveniently used to solve bag-of-words model. The surprise can be simplified when the words in a document be treated as Bernoulli random variables. This reduces the time complexity to compute the surprise of a document to the order of document size. This reduction in complexity enables the IXS derived samplers remain fast and scalable while exploring new information.



We now present our analysis from synthetic networks. An extensive set of experiments on synthetic attributed networks reveal the effect of network structure, content purity (separability of clusters), skewness of cluster sizes and content-network assortativity on sampling performance. First, we observe that link-trace samplers show surprisingly little variation over different network structures including Barabasi, Watts-Strogatz and LFR \cite{lancichinetti2009benchmarks} network structures.  Second, purity or separability of content clusters enable samplers discover clusters much faster. Third, skewness of the cluster sizes makes it difficult for smaller sized clusters to be discovered. Fourth, assortativity controls the access of link-trace samplers to different clusters; therefore lower assortativity means better cluster access and discovery. Finally, we show via Lemma~\ref{lemma:coverage} that \textit{unseen attribute} biasness and \textit{stratified content} sampling helps attribute-aware samplers such as IXS circumvent the above effects.


Now, we begin with a detailed discussion of results for synthetic networks. 

\textit{Network effect: } Link trace samplers show surprisingly little variation over very different network structures. In this experiment, we tested different networks structures : Barabasi, Watts-Strogatz ($p = 0.1$, high clustering coefficient) and LFR ($\mu = 0.1$, high clustering with power law degree distribution). We notice that the samplers' relative performance remains unchanged over the three networks. Furthermore, we observe that IXS based samplers does consistently better than other samplers. IXS stratifies the content by sampling nodes that have very different content characteristics, and this makes the content easier to cluster. It is interesting to note that attribute-aware samplers outperform UNI since independent sampling of network nodes is considered as the ideal case for content. UNI over-samples nodes in clusters of large sizes, thus making them poorer at cluster preservation. Attribute-agnostic samplers are affected not only by the cluster size but also by the assortativity in network making them even poorer at cluster preservation.

%
%

\begin{table*}[!htb]\footnotesize\mdseries
\centering
\caption{The clustering preservation goal is to have maximum possible normalized mutual information. The table shows different samplers' performances at preserving content-clusters in a LFR ($\mu =0.1$) network. It is evident from the table that the parameters of skew ($s_l=0,\ s_m=0.22,\ s_h=0.52$), purity ($p_l=0.2,\ p_h=10$) and assortativity ($a_l=0,\ a_h=1$) have a significant impact on the classification performance at different values ($l$: low, $m$: medium, $h$: high). The last row depicts the performance improvement of IXS over UNI. This is prominent at high skew, high purity and low assortativity.}
\label{table:cluster}
\resizebox{\textwidth}{!}{%
\begin{tabular}{@{}c cc r cc r  cc r cc r cc r cc@{}}
 \toprule
{} & \multicolumn{5}{c}{$s_l$} & \phantom{a} & \multicolumn{5}{c}{$s_m$} & \phantom{a}& \multicolumn{5}{c}{$s_h$}\\
\cmidrule{2-6} \cmidrule{8-12} \cmidrule{14-18}
{} & \multicolumn{2}{c}{$a_l$} & \phantom{a} & \multicolumn{2}{c}{$a_h$} & \phantom{a} & \multicolumn{2}{c}{$a_l$} & \phantom{a} & \multicolumn{2}{c}{$a_h$} & \phantom{a} & \multicolumn{2}{c}{$a_l$} &\phantom{a} & \multicolumn{2}{c}{$a_h$}\\
\cmidrule{2-3} \cmidrule{5-6} \cmidrule{8-9} \cmidrule{11-12} \cmidrule{14-15} \cmidrule{17-18}
samplers & $p_l$ & $p_h$ && $p_l$ & $p_h$ && $p_l$ & $p_h$ && $p_l$ & $p_h$ && $p_l$ & $p_h$ && $p_l$ & $p_h$\\ \midrule
BFS&0.42&0.996&&0.414&0.983&&0.384&0.958&&0.405&0.933&&0.31&0.882&&0.318&0.84\\ 
RW&0.421&0.995&&0.42&0.992&&0.391&0.961&&0.38&0.945&&0.328&0.931&&0.305&0.848\\ 
MHRW&0.426&0.995&&0.419&0.993&&0.378&0.961&&0.395&0.946&&0.328&0.921&&0.295&0.846\\ 
FF&0.418&0.995&&0.421&0.99&&0.378&0.959&&0.391&0.942&&0.315&0.91&&0.306&0.827\\ 
XS&0.419&0.995&&0.422&0.992&&0.382&0.965&&0.385&0.936&&0.313&0.953&&0.309&0.839\\ \midrule
UNI&0.417&0.996&&0.416&0.994&&0.383&0.961&&0.393&0.956&&0.318&0.862&&0.304&0.909\\ \midrule
ExP&\textbf{0.607}&0.961&&\textbf{0.599}&0.964&&\textbf{0.581}&0.946&&\textbf{0.573}&0.948&&\textbf{0.518}&0.95&&\textbf{0.522}&0.95\\ 
BAL&0.423&\textbf{0.999}&&0.425&0.995&&0.393&\textbf{0.992}&&0.389&\textbf{0.991}&&0.335&\textbf{0.979}&&0.337&\textbf{0.987}\\ 
IXS&0.429&0.998&&0.416&\textbf{0.996}&&0.393&0.987&&0.392&\textbf{0.991}&&0.318&0.972&&0.32&0.981\\ 
H-IXS&0.501&0.989&&0.494&0.991&&0.494&0.98&&0.483&0.983&&0.457&0.977&&0.464&0.979\\ \midrule
$\frac{IXS-UNI}{UNI}$&4.28\%&0.20\%&&0.00\%&0.20\%&&2.61\%&2.71\%&&-0.25\%&3.66\%&&0.00\%&12.76\%&&5.26\%&7.92\%\\ \bottomrule
 \end{tabular}
 }
\end{table*}

\textit{Content dependence: } Attribute-aware samplers perform significantly better than attribute-agnostic samplers. The results as shown in \Cref{table:cluster} reveal that for every case, there is a attribute-aware samplers that outperforms other samplers in a statistically significant manner ($p < 0.05$). In the best scenario of medium skew, high assortativity and low purity ($s_m, a_h, p_l$), attribute-aware samplers achieves an improvement over existing baselines by as much as 45\%.

Content characteristics---purity, skew and assortativity---affect the sampling performance to varying extent. \Cref{table:cluster} suggests that each of the three characteristics has different impact on clustering performance. At low purity, high skew and low assortativity, the samplers are least likely to discover clusters accurately, while at high purity, low skew and high assortativity, the samplers are very efficient at preserving clusters.

\textit{Purity: }All samplers perform better with increased purity. As purity increases, the clusters get well separated thereby improving the performance. It is therefore not surprising that when purity is high, skew and assortativity becomes less relevant. At the same value of purity, samplers other than content aware samplers, show a marginal decrease in performance as skew and assortativity increase.

\textit{Skew: }Observing the three blocks of columns in \Cref{table:cluster} from left to right, we clearly see that increased skew makes it harder for the samplers to preserve the clustering. It is unclear as to why the performance slightly improves at mid-skew levels which was expected to be slightly worse off. This may be due to an interaction effect among the three parameters and is an effect not fully understood. In general, as skew in the data cluster increases, it becomes increasingly difficult to identify smaller sized clusters, thereby degrading the clustering performance.

\textit{Assortativity: }High values of assortativity cause performance degradation at mid-skew and high-skew levels in \Cref{table:cluster} by acting as a bottleneck for link trace samplers. Observe that at mid-skew and high-skew levels and high-purity levels, greater randomness ($a_l \rightarrow 0$) leads to greater correct ($p_h$) cluster information and therefore better performance. Conversely at low-purity levels, NMI performance degrades owing to the fact that now greater information means more noisy information since clusters overlap. This happens due to the fact that high purity implies clusters are well separated and therefore learning about new clusters improves the overall clustering performance. In sum, assortativity controls access of link-trace samplers to different content clusters.

In this section, we observe that higher-level content structure, i.e. content clusters. Importantly, we find that attribute-aware sampling is better at preserving cluster information than uniform sampling. This is important since uniform samplers are the gold standard for sampling attributes. In the next section, we explore the impact of sampling on another data mining task---classification.


\section{Classification}
\label{sec:Classification}
\begin{table*}[!htb]\footnotesize\mdseries
\centering
\caption{The classification preservation goal is to have maximum possible weighted $F_1$ score. The table shows different samplers' performances at predicting hidden cluster-\texttt{id} on a LFR ($\mu =0.1$) network. It is evident from the table that the parameters of skew ($s_l=0,\ s_m=0.22,\ s_h=0.52$), purity ($p_l=0.2,\ p_h=10$) and assortativity ($a_l=0,\ a_h=1$) have a significant impact on the classification performance at different values ($l$: low, $m$: medium, $h$: high). The last row depicts the performance improvement of IXS over UNI. This is prominent at high skew, high purity and low assortativity.}
\label{table:classification}
\resizebox{\textwidth}{!}{%
\begin{tabular}{@{}c cc r cc r  cc r cc r cc r cc@{}}
 \toprule
{} & \multicolumn{5}{c}{$s_l$} & \phantom{a} & \multicolumn{5}{c}{$s_m$} & \phantom{a}& \multicolumn{5}{c}{$s_h$}\\
\cmidrule{2-6} \cmidrule{8-12} \cmidrule{14-18}
{} & \multicolumn{2}{c}{$a_l$} & \phantom{a} & \multicolumn{2}{c}{$a_h$} & \phantom{a} & \multicolumn{2}{c}{$a_l$} & \phantom{a} & \multicolumn{2}{c}{$a_h$} & \phantom{a} & \multicolumn{2}{c}{$a_l$} &\phantom{a} & \multicolumn{2}{c}{$a_h$}\\
\cmidrule{2-3} \cmidrule{5-6} \cmidrule{8-9} \cmidrule{11-12} \cmidrule{14-15} \cmidrule{17-18}
samplers & $p_l$ & $p_h$ && $p_l$ & $p_h$ && $p_l$ & $p_h$ && $p_l$ & $p_h$ && $p_l$ & $p_h$ && $p_l$ & $p_h$\\ \midrule
BFS&0.302&0.917&&0.306&0.870&&0.269&0.824&&0.261&0.783&&0.202&0.750&&0.196&0.619\\
RW&0.300&0.907&&0.311&0.891&&0.274&0.817&&0.268&0.781&&0.203&0.746&&0.198&0.628\\
MHRW&0.294&0.906&&0.312&0.876&&0.269&0.803&&0.264&0.772&&0.204&0.728&&0.192&0.652\\
FF&0.295&0.915&&0.304&0.907&&0.269&0.820&&0.268&0.779&&0.198&0.747&&0.198&0.615\\
XS&0.298&0.920&&0.301&0.883&&0.265&0.822&&0.268&0.745&&0.198&0.807&&0.201&0.577\\ \midrule
UNI&0.297&0.921&&0.301&0.905&&0.268&0.788&&0.263&0.768&&0.203&0.684&&0.200&0.700\\ \midrule
ExP&0.314&0.784&&0.317&0.731&&\textbf{0.332}&0.820&&\textbf{0.312}&0.774&&\textbf{0.310}&0.805&&\textbf{0.304}&\textbf{0.860}\\
BAL&\textbf{0.377}&\textbf{0.968}&&\textbf{0.393}&0.962&&0.317&\textbf{0.904}&&0.306&\textbf{0.882}&&0.224&0.892&&0.216&0.804\\
IXS&0.376&\textbf{0.968}&&0.392&\textbf{0.965}&&0.318&0.902&&0.306&0.880&&0.224&\textbf{0.893}&&0.216&0.804\\ \midrule
$\frac{IXS - UNI}{UNI}$&26.6\%&5.1\%&&30.2\%&6.6\%&&18.7\%&14.5\%&&16.3\%&14.6\%&&10.3\%&30.6\%&&8.0\%&14.9\%\\ \bottomrule

 \end{tabular}
 }
\end{table*}

In this section, we present our results for classification. We shall first describe the experimental setup used for the experiments, followed by presentation of results.

A classifier is defined as follows. Assume that we have a collection of $n$ objects with corresponding features \{$x_i$\} and target label \{$y_i$\}. The goal is to learn a function $f(x)$ using the given collection that predicts a label $y$ for an unseen input $x$. Thus, the classification goal is to learn $f$ given \{ $x_i$, $y_i$\} such that $\mathbb{E}(||f(x) - y||)$ is minimized over unseen inputs $x$. Furthermore, we assume graph input for training the classifier and the target labels are available.

\subsection{Experimental setup}
\label{sub:Experimental setup2}
\setlength{\belowcaptionskip}{-10pt}
\begin{figure}[b]
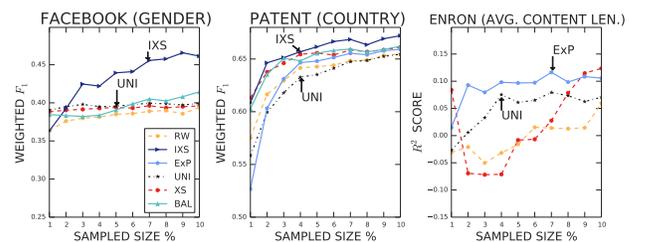

  \centering
\adjincludegraphics[width=\linewidth, trim={2cm 0cm 2cm 0cm},clip]{{real_cluster/fb_pat_enron_class}.pdf}
\caption{Classification performance on three real world datasets with their corresponding predicted attribute mentioned in the title. The classification performance on Facebook, Patent and Enron shows that IXS and its variants are better than attribute-agnostic samplers..}
\label{fig:real_classes}
\end{figure}

In this section, we provide a description of the experimental setup comprising a view of datasets used, evaluation methodologies and target label selection for classification.

\textit{Dataset description: } We test preservation of content classes on real-world and synthetic datasets. For real-world datasets, we employ three networks---Facebook, Patent and Enron. We did not test our classifier on Pokec that has too few (just two) categorical attributes, and Wikipedia that has too many (7,969) categorical attributes, in absence of a relevant prediction task. We use one of the attributes as the target label ($y$) and the rest of attributes as the features ($x$). For the synthetic datasets, we use different variations of the network generators, and different combinations of content characteristics by varying the purity, skew and assortativity~\citep{Kumar2016}. Furthermore, we use three attributes of nodes: two continuous and one discrete attribute as the features ($x$) and cluster \texttt{id} as the target ($y$).  Additionally, we normalize all continuous and discrete features using well-known z-score standardization and one-hot encoding respectively \citep{beck2000high}.

\textit{Evaluation :} The evaluation metric for classification differ slightly depending on whether the target attribute ($y$) is discrete or continuous. For predicting discrete attributes, we use the weighted $F_1$ score while $R^2$ coefficient of determination for predicting continuous attributes.

\textit{Target label selection: }We choose specific attributes to predict for real-world and synthetic datasets. The choice for picking attributes to predict for each of the real-world datasets---Facebook, Patent and Enron---was based on the following two principles. First, across the three real-world datasets, we picked attributes of varying cardinality to help us understand the effect of cardinality. Second, we employed a simple principle to identify variables used to predict the target attribute: the target attribute had to exhibit correlation to the features attributes. We note in passing that feature selection itself is often a hard task. Thus, for Facebook we predict ``gender'' using feature set {``locale'', ``education type''}. For Patent dataset, we predict attributes ``country'' in Patent and ``Average content length'' in Enron using the rest of the attributes as features. Finally for synthetic dataset, we choose to predict the cluster \texttt{id}.

\subsection{Experimental results}
\label{sub:Experimental results2}

Now, we present the experimental results of real-world and synthetic networks. We begin the results with real-world networks.

We present the classification results for real world datasets: Facebook, Patent and Enron. Attribute-aware samplers are a better choice of sampling than UNI for all classification tasks. The classification performance for Facebook, Patent and Enron (regression) is shown in \Cref{fig:real_classes}. For Facebook dataset, IXS achieves over 18\% relative gain in the weighted $F_1$ over competing samplers such as UNI and RW variants. For Patent dataset, we note that attribute-aware and attribute-agnostic samplers are better than UNI by a margin of 2\%. The overall weighted $F_1$ performance of almost all samplers is high due to skewed distribution of the target attribute (i.e. country's skew = 0.70). Similarly for Enron dataset, we observe the ExP is a better choice for sampling than RW or UNI sampler. UNI and other attribute-agnostic samplers suffer from class skew in most real-world datasets. However, attribute-aware samplers such as IXS and ExP circumvent this problem by effective diversified or surprise-driven sampling in attribute value space.

We now present classification analysis on synthetic networks. Classification and clustering results are remarkably similar for synthetic dataset. Since, our target class and clusters are exactly the same, the results are therefore very similar. From \Cref{table:classification}, we observe that for any combination of content characteristics including purity, skew and assortativity, attribute-aware samplers are significantly better than attribute-agnostic samplers. It is not surprising that IXS and other attribute-aware samplers consistently outperform UNI. This is because IXS to some extent solves the ``class balancing problem'' via stratified sampling of objects from each class. IXS achieves a performance gain of much as 30\% over the uniform baseline samplers. Furthermore, IXS does best when the classes are highly skewed, the attributes have low assortative mixing and the purity of classes is high.

\section{Related Work}
\label{sec:relatedwork}

Not surprisingly, inference from content gathered by network sampling arises in many diverse areas. We study the prior works related to our work in three different areas:  network sampling, content sampling and joint network-content sampling


``Representative subgraph'' sampling closely resembles to our methodology of sampling attributed network. Representative subgraph sampling aims to construct a sampled subgraph that has \textit{network structure} very similar to that of the original network. Forest Fire~\cite{leskovec2006sampling}  preserved several key network structure characteristics. Hubler et al. \citep{Hubler2008} showed via Metropolis algorithm that prior knowledge of the network can help in obtaining better representative samples. The objective of our work is preservation of content properties and not just the network structure.

Another line of research on network sampling focuses on understanding the biases of existing samplers and ways to obtain \textit{uniform samples}. Kurant et al. \citep{kurant2010bias}  quantified the degree bias for several network samplers and proposed new ways to correct them. Costenbander et al. \citep{costenbader2003stability} did thorough analysis of the effect of noise (sample) on network centrality estimation. Gjoka et al. \citep{Gjoka2010} implemented the proposed uniform link-trace samplers on very massive Facebook network to validate the results. Chiericetti et al. \citep{chiericetti2016sampling} proposed an efficient random walk sampling strategy for sampling according to a prescribed distribution not just ``uniform'' sampling. Maiya et al. \citep{maiya2011benefits} exploited the bias instead of correcting it to design expansion based samplers. Similar to Maiya's work, we exploit the bias of entropy based samplers to balance the attributes in the sample, yielding improved classification and clustering performance.

There has been a plethora of research on sampling content from an \textit{unknown population distribution}. Our objective resembles with these surveys that try to estimate the underlying content characteristics. However most the well known samplers such as Poisson sampling, stratified sampling, etc. \citep{patton2005qualitative} require random access to the nodes in dataset, prior knowledge in some cases, and therefore fail to capture network structure. The idea of surprise based data sampling is however not new. It has been used in the fields of graph visualization, information retrieval, active learning, etc. For example, in classical database search, Sarwagi \citep{sarawagi2000user} used the Maximum Entropy principle to model a user's knowledge and aid the user in exploring OLAP data cubes. In graph visualization work \cite{pientaadaptivenav}, the authors chose to highlight the neighbors that are most surprising in information. Our work borrows the idea of surprise defined in terms of entropy and stratified sampling principles to design better attributed-network samplers.

Sociological and statistical studies on social networks such as friendship recommendation, link prediction, attribute inference, type distribution, etc. implicitly rely upon both content and network. However there is little prior work on understanding the effect of sampling on \textit{joint network-content characteristics}. Li et al. \cite{li2011sampling} studied five different sampling strategies for node-type and link-type distribution preservation. They noted that sample size of 15\% from RDS, the best sampling strategy, can preserve ``location type'' distribution in Twitter network very well. Yang et al. \citep{yang2013semantically} proposed a semantic sampling strategy, Relational Profile sampling, that preserves the semantic relationship types in a heterogeneous networks. Park et al. \cite{park2013sampling} remarked about the inefficiency of the existing network samplers in estimating node attributes. Although seemingly similar, the previous works have been specific to tasks such as attribute distribution and node-type preservation. However, we present samplers for tasks such as clustering and classification, along-with theoretical proofs of bias and convergence. To the best of our knowledge, we are the first to propose network samplers for data mining purposes like clustering and classification.

\vspace{-5pt}
\section{Limitations}
\label{sec:Limitations}


Now, we discuss limitations of this work. First, much of the analysis assumes that we have no missing values; while the algorithms would work in the case of missing values, it would useful to introduce a noise model to formally estimate error in surprise when confronted with missing values. Second, the time and space complexity of IXS is greater than MHRW and RW. The incremental update complexity is $O(\mu \log |\mathbb{S}| + \mu^2)$, where $\mu$ is the mean degree, while it is $O(1)$ for RW or MHRW. Some of this can be mitigated by appropriate content-network structures. For example, we commonly assume that a node has access to the \texttt{id}'s of its neighbors, but not the attributes of its neighboring nodes; this can be easily rectified, reducing the incremental time complexity. Third, the analysis of the Wikipedia dataset reveals that when the attributes are numerous, binary and highly asymmetric, we need to modify the definition of surprise to handle the case, perhaps by defining an symmetric version of surprise. Finally, our model of link-trace sampling is limited: many social networks allow us to make queries on the content directly by returning network nodes that satisfy the query. It would be interesting to expand the sampling paradigm to incorporate a more rich query model.

\begin{table*}[b]
\center
\caption{ The network-structure preservation goal is to obtain the least possible KS statistic of degree, clustering coefficient and path length distributions. Attribute-agnostic samplers such as XS, FF and RW are tuned to preserve network characteristics, perform better than attribute-aware samplers. We show the network characteristics for three datasets--Facebook, Pokec and Philosopher. }\label{table:remaining}
\resizebox{\textwidth}{!}{%
\begin{tabular}{@{}c ccc ccc ccc c@{}}
 & \multicolumn{3}{c}{\textit{Facebook}} & \multicolumn{3}{c}{\textit{Pokec}} & \multicolumn{3}{c}{\textit{Wikipedia}} \\ \toprule
Samplers& Degree& CC& Path& Degree& CC& Path& Degree& CC& Path\\ \toprule
BFS&0.403&0.295&0.781&\textbf{0.123}&\textbf{0.118}&0.2765&0.362&0.338&0.552\\ 
RW&\textbf{0.334}&0.257&0.636&0.235&0.157&\textbf{0.2745}&0.332&0.216&0.405\\ 
MHRW&0.349&\textbf{0.209}&0.461&0.371&0.239&0.491&0.35&0.212&0.232\\ 
FF&0.336&0.234&0.546&0.246&0.168&0.3085&0.335&0.211&0.336\\ 
XS&0.666&0.522&\textbf{0.1}&0.219&0.271&0.766&\textbf{0.3}&\textbf{0.209}&0.396\\ \midrule
UNI&0.756&0.623&0.865&0.645&0.595&0.976&0.616&0.553&0.78\\ \midrule
ExP&N/A&N/A&N/A&N/A&N/A&N/A&N/A&N/A&N/A\\ 
BAL&0.517&0.299&0.461&0.513&0.304&0.588&0.398&0.225&\textbf{0.181}\\ 
IXS&0.454&0.424&0.326&0.518&0.26&0.631&0.486&0.414&0.626\\ \bottomrule
\end{tabular}%
}\end{table*}

\section{Conclusion}
\label{sec:conclusion}

In this paper, we have presented new attribute-aware sampling methodologies for attributed networks. The problem is important because data mining tasks such as clustering or classification are commonplace on the nodal attributes of real-world networks. A key challenge is that these large networks are often sampled with BFS or RW, which were never designed to preserve content characteristics. In the first of its kind study, we show that these samplers are suboptimal for standard data mining tasks. We proposed several samplers based on the idea of information expansion. We have excellent results with information based sampling outperforming the baselines for data mining tasks such as cluster preservation, with average-case performance improvements over 45\%.


\bibliographystyle{ACM-Reference-Format}
\bibliography{sigproc}

\end{document}